\DeclareSymbolFont{AMSb}{U}{msb}{m}{n}
\documentclass[reqno,centertags,11pt,a4paper,noamsfonts]{amsart}
\pdfoutput=1
\usepackage[svgnames]{xcolor}
\usepackage[T1]{fontenc}
\usepackage{lmodern}
\usepackage{amsmath,amssymb,mathtools,mathrsfs}
\usepackage{enumitem}
\usepackage{xcolor}
\usepackage{hyperref}

\usepackage[margin={3cm},marginparwidth={2.5cm}]{geometry}

\usepackage[
  backend=biber,
  style=numeric,
  sorting=none,
  giveninits=true,
  doi=true,
  maxbibnames=6,
  isbn=false,
  url=true,
  eprint=false,
  sorting=nyt
]{biblatex}

\DeclareFieldFormat[article]{volume}{\mkbibbold{#1}}
\DeclareFieldFormat[article]{number}{\mkbibparens{#1}}

\renewbibmacro*{volume+number+eid}{%
  \printfield{volume}%
  \setunit*{\addspace}%
  \printfield{number}%
  \setunit{\addcomma\space}%
  \printfield{eid}%
}
\renewbibmacro{in:}{%
  \ifentrytype{article}
    {}
    {\printtext{\bibstring{in}\intitlepunct}}}

\AtEveryBibitem{%
  \clearfield{note}%
}

\renewbibmacro*{doi+eprint+url}{%
  \printfield{doi}%
  \newunit\newblock
  \iffieldundef{doi}
    {\usebibmacro{url+urldate}}
    {}%
}

\newcommand{\C}{\mathcal{C}}
\newcommand{\1}{\mathbf{1}}
\newcommand{\ii}{\imath}
\newcommand{\dd}{\mathrm{d}}

\renewcommand{\Re}{\mathrm{Re}\,}

\definecolor{light_gray}{gray}{0.75}
\definecolor{lighter_gray}{gray}{0.5}
\colorlet{light_blue}{blue!20}
\definecolor{dark_green}{rgb}{0.0, 0.6, 0.0}
\definecolor{royal_blue}{rgb}{0.0, 0.22, 0.66}
\definecolor{salmon}{rgb}{1.0, 0.55, 0.41}
\definecolor{gold}{rgb}{0.8, 0.63, 0.21}
\definecolor{navy_blue}{rgb}{0.0, 0.0, 0.5}
\definecolor{crimson}{rgb}{0.79, 0.0, 0.09}
\definecolor{amethyst}{rgb}{0.6, 0.4, 0.8}
\definecolor{alizarin}{rgb}{0.82, 0.1, 0.26}
\definecolor{amaranth}{rgb}{0.9, 0.17, 0.31}
\definecolor{azure}{rgb}{0.0, 0.5, 1.0}
\definecolor{canaryyellow}{rgb}{0.82, 0.41, 0.12}
\definecolor{carrotorange}{rgb}{0.8, 0.33, 0.0}
\definecolor{cadmiumgreen}{rgb}{0.0, 0.42, 0.24}
\definecolor{copper}{rgb}{0.72, 0.45, 0.2}
\definecolor{aqua}{rgb}{0.5, 1.0, 0.83}
\definecolor{awesome}{rgb}{1.0, 0.13, 0.32}
\definecolor{candyapplered}{rgb}{1.0, 0.03, 0.0}
\definecolor{caribbeangreen}{rgb}{0.0, 0.8, 0.6}
\definecolor{indigo}{rgb}{0.0, 0.25, 0.42}

\DeclareMathOperator{\weaklystar}{\rightharpoonup\kern-2.2ex ^* \, \,}

\def\XXint#1#2#3{{\setbox0=\hbox{$#1{#2#3}{\int}$ }
\vcenter{\hbox{$#2#3$ }}\kern-.6\wd0}}

\newcommand{\R}{\mathbb R}
\newcommand{\N}{\mathbb N}
\newcommand{\Z}{\mathbb Z}
\renewcommand{\C}{\mathbb C}

\newcommand\norm[1]{\lVert #1 \rVert}

\DeclareMathOperator{\rank}{rank}

\DeclareMathOperator{\Tr}{Tr}

\newcommand{\mL}{\mathrm{L}}
\renewcommand{\phi}{\varphi}

\newcommand{\mH}{\mathrm{H}}

\newcommand{\T}{\mathbb{T}}
\newcommand{\ee}{\mathrm{e}}

\renewcommand{\geq}{\geqslant}\renewcommand{\leq}{\leqslant}

\newcommand\cD{\mathcal{D}}

\newcommand\cS{\mathcal{S}}

\newcommand\cY{\mathcal{Y}}

\DeclareMathOperator{\tr}{Tr}
\DeclareMathOperator{\ran}{Ran}

\usepackage{dsfont} 
\def\1{{\mathds{1}}}

\newcommand\inner[1]{\langle #1 \rangle}

\newcommand{\acs}[1]{\left\{ #1 \right\}}                       
\newcommand{\ab}[1]{\left|#1\right|}                                
\newcommand{\nor}[2]{ \left| \! \left| #1 \right| \! \right|_{#2} } 

\newcommand{\mysup}[1]{\underset{\substack{#1}}{\text{\normalfont{sup}}}\;} 

\newcommand{\ketbra}[2]{\left| #1 \right\rangle \left\langle #2 \right|}

\newcommand{\rhor}[1]{\rho^\R_{#1}}
\newcommand{\rhot}[1]{\rho^\T_{#1}}

\theoremstyle{plain}
\newtheorem{theorem}{Theorem}[section]
\newtheorem{proposition}[theorem]{Proposition}

\newtheorem{corollary}[theorem]{Corollary}
\newtheorem{lemma}[theorem]{Lemma}

\newtheorem*{theorem*}{Theorem}

\theoremstyle{definition}

\newtheorem{remark}[theorem]{Remark}
\newtheorem*{remark*}{Remark}

\hypersetup{
	linktoc=all,
	bookmarksdepth=3,
	colorlinks=true,
	linkcolor=Green,citecolor=Orange,urlcolor=DarkBlue,
	pdftitle={Sharp one-dimensional gradient-correct Lieb-Thirring inequalities and inverse Kohn-Sham theory},
	pdfauthor={Thiago Carvalho Corso, Louis Garrigue}, 
	pdfsubject={},
	pdfkeywords={}} 
\begin{document}
\numberwithin{table}{section}

\title[A sharp gradient-corrected LT inequality in one dimension and applications]{A sharp gradient-corrected Lieb--Thirring inequality in one dimension with applications to density functional theory}

\author[T.~Carvalho~Corso]{Thiago Carvalho Corso*}
\address[T.~Carvalho Corso]{Institute of Applied Analysis and Numerical Simulation, University of Stuttgart, Pfaffenwaldring 57, 70569 Stuttgart, Germany}
\email{thiago.carvalho-corso@mathematik.uni-stuttgart.de}

\author[L.~Garrigue]{Louis Garrigue*}
\address[L.~Garrigue]{Laboratoire ``analyse géométrie modélisation'', CY Cergy Paris Université, 95302 Cergy-Pontoise, France}
\email{louis.garrigue@cyu.fr}


\date{\today}
\thanks{\emph{Funding information}:  DFG -- Project-ID 572811220.\\[1ex]
\textcopyright 2026 by the authors. Faithful reproduction of this article, in its entirety, by any means is permitted for noncommercial purposes.}
\begin{abstract}
	We first prove a sharp Lieb-Thirring inequality with semi-classical constant and gradient correction in one dimension. In the periodic setting, the corresponding inequality also includes a parity dependent shell-correction term, which appears to be new. Moreover, we also characterize the optimizers of these inequalities. In particular, the optimizing density profiles form a two-dimensional manifold, which we call here the Poincar\'e density manifold. 

We then apply these results to study the Thomas--Fermi--von Weizs\"acker (TFW) approximation of the kinetic energy in Density Functional Theory. More precisely, we first obtain an explicit and non-asymptotic bound on the global and local error made by the TFW approximation. These bounds are controlled by the distance of the normalized density profile from the Poincar\'e density manifold times the number of particles. In particular, these bounds show that the TFW approximation is exact on the Poincar\'e manifold. We then show that, in the periodic setting, the distance between the Kohn-Sham potential and the inverse TFW potential in the natural dual Sobolev space can also be controlled by explicit density dependent constants times the distance to the Poincar\'e manifold. As a corollary, we prove that for a fixed density profile, the Kohn-Sham potential converges towards the inverse Thomas--Fermi potential in the high-density limit $N \rightarrow \infty$ with asymptotic rate of order $N^{-2}$. This rate is shown to be optimal.
\end{abstract}
\setcounter{tocdepth}{1}
\maketitle
\setcounter{secnumdepth}{3}

\section{Introduction}
\label{sec:introduction}

\subsection{Motivation}

Density functional theory (DFT) is one of the most widely used methods for the calculation of electronic structure properties of molecules and materials. In its ground-state formulation, DFT seeks to replace the electronic many-body problem by a variational problem for the single-particle density.  While this reformulation is formally exact, its practical value depends on finding accurate approximations which are explicit, or at least computationally tractable, in terms of the density \cite{HohKoh64,Levy79,ParrYang89,Becke14}. This issue is particularly pronounced in orbital-free DFT, where density-explicit approximations are needed not only for the exchange-correlation energy but also for the kinetic energy \cite{WangTeter92,WangGovindCarter98,MiLuoTrickeyPavanello23,OkunCancioBurke24}. 

The simplest such approximation is provided by the Thomas--Fermi (TF) theory, which approximates the exact kinetic energy by the local kinetic energy of a homogeneous Fermi gas evaluated at the local density.  Despite its simple local form, Thomas--Fermi theory has a firm mathematical foundation in several asymptotic regimes \cite{LiebSimon73,LiebSimon77,Lieb81,NamPortmannSolovej2012,FouLewSol18,GottschlingNam2018}. A closely related but non-asymptotic expression of the same kinetic energy principle is provided by the Lieb--Thirring (LT) inequality \cite{LiebThirring75,LiebThirring76}. Indeed, in its kinetic formulation, the LT inequality gives a universal lower bound on the fermionic kinetic energy in terms of the local power of the density appearing in the TF theory. In particular, the LT inequality provides a universal density-explicit constraint and consequently a useful benchmark for density-explicit approximations of the kinetic energy.

However, determining the optimal constant in the LT inequality, and therefore the strongest possible lower bound of this form, is a longstanding problem in mathematical physics. The Lieb--Thirring conjecture predicts the optimal constant to come from either the one-particle regime or the semiclassical regime, with the latter corresponding to the Thomas--Fermi constant. Although substantial progress has been made \cite{Weidl96,IlyinLaptevZelik20,FrankHundertmarkJexNam21,FrankGontierLewin21,HundertmarkKunstmannRiedVugalter23,CorsoRied25,Cor26e,Read26}, this problem remains open in the physically most relevant setting of three-dimensional systems \cite{Frank21,Schimmer22,Nam23,FrankLaptevWeidl22}. Furthermore, even with the optimal constant, the classical Lieb--Thirring inequality provides a purely local lower bound and is therefore insensitive to the spatial scale on which the density varies. Indeed, equimeasurable density profiles give the same LT lower bound even if one varies much more rapidly than the other. In contrast, derivative-based quantities can distinguish such profiles and thus provide a natural way to incorporate information about spatial inhomogeneity. This motivates refinements of local kinetic-energy bounds with gradient-correction terms.

Such gradient corrections play an important role in practical density-functional approximations. The simplest example is the von Weizs\"acker term \cite{vonWeizsacker35,BenguriaBrezisLieb81,Lieb81}, while more general density-gradient corrections form the basis of semilocal approximations such as the generalized gradient approximations \cite{Becke88,PerdewWang91,PBE96}. Beyond these smooth corrections, however, finite systems also exhibit contributions arising from the discreteness of the spectrum, commonly referred to as shell or finite-size effects \cite{YannouleasLandman93,YannouleasLandman13}. For instance, recent work by Burke and collaborators \cite{ElliottLeeCangiBurke08,RibeiroBurke17,RibeiroBurke18,Burke20,Burke20Faraday} has clarified how semiclassical expansions, gradient expansions, and finite-system shell corrections combine to determine the next-order structure beyond the leading local density approximation of the kinetic energy. Nevertheless, from a mathematical perspective, our understanding is much less complete at this level: while the leading local-density approximation has by now received a substantial rigorous foundation even for interacting systems \cite{LewinLiebSeiringer20,LLSUniversal,Mietzsch20}, rigorous results that identify and control the next-order corrections remain scarce. In particular, to the best of our knowledge, there are currently no rigorous density-explicit estimates that simultaneously capture the gradient and finite-size shell corrections to local density approximations.

It is therefore the purpose of this work to derive such estimates for the kinetic energy in the simplified setting of one-dimensional fermions. 

\subsection{Main contributions} The main contributions of this work can be summarized as follows:
\begin{enumerate}
\item We prove sharp gradient-corrected Lieb--Thirring inequalities on the real line and on the one-dimensional torus. On the torus, the inequality also includes a shell term. We also explicitly characterize all density matrices and the corresponding single-particle densities that saturate these inequalities.

\item Using this inequality, we derive non-asymptotic error bounds for the Thomas--Fermi--von Weizs\"acker (TFW) approximation to the exact kinetic-energy functional. Moreover, we obtain corresponding local estimates comparing the exact kinetic-energy density with its TFW approximation. These estimates are controlled by the distance of the single-particle density from the manifold of optimizers of the gradient-corrected LT inequality. In particular, they show that the TFW approximation is \emph{exact} on the manifolds of optimizers for any number of particles. 

\item We further use these results to obtain error bounds between the exact inverse Kohn--Sham potential and its TFW counterpart. As a consequence, we prove that, for a fixed normalized density profile, the appropriately rescaled inverse Kohn--Sham potential converges to its Thomas--Fermi counterpart at rate $N^{-2}$ as the particle number $N\rightarrow\infty$. This rate is also shown to be optimal.
\end{enumerate}

Taken together, these results provide a rigorous quantitative justification of an explicit gradient-corrected approximation in one-dimensional density-functional theory.

\section{Main results}

In this section we first state the main results of the paper in detail. We then present some applications and outline the main arguments in the proofs. Finally, we compare these results with related results in the literature.

Throughout the paper, for any $N \in \R_+$ and $\Omega = \R$ or $\Omega = \T$, we define the set of one-particle density-matrices
\begin{align}\label{eq:def_dm} 
	\begin{aligned}
	\mathfrak{S}(\Omega) &:= \acs{\gamma \in \mathcal{B}(L^2(\Omega)) \;|\; 0 \leq \gamma \leq 1, \quad \tr \gamma < \infty, \quad \tr (-\partial_x^2) \gamma < +\infty }, \\
	\mathfrak{S}_N(\Omega) &:= \acs{\gamma \in \mathfrak{S}(\Omega) \;|\; \tr \gamma = N}
	\end{aligned}
\end{align}
and moreover,
\[
    T[\gamma]:=\Tr(-\partial_x^2\gamma),
    \qquad
    G(\rho):=\int_\Omega \ab{(\sqrt\rho)'}^2 = \frac14 \int_\Omega \frac{|\rho'|^2}{\rho},
    \qquad
    \text{and}\quad c_{\rm TF}:=\frac{\pi^2}{3}.
\]

\subsection{Gradient-corrected Lieb--Thirring inequalities}

Our first main result is a sharp Lieb-Thirring inequality with gradient correction in $\R$.

\begin{theorem}[Sharp gradient-corrected Lieb--Thirring inequality on the line]
\label{thm:LT_line}
Let $N \in \N$ and $\gamma \in \mathfrak{S}_N(\R)$ be a one-particle density matrix on $L^2(\R)$, and define $\rho := \rho_\gamma$. Then
\begin{align}
    T[\gamma]+3G(\rho)
    \geq
    \frac{\pi^2}{3}\left(1+\frac2{N^2}\right)
    \int_\R\rho^3.
    \label{eq:LT_line_main}
\end{align}
Moreover, equality is attained if and only if there exists $\mu > 0$ such that up to translation,
\begin{align}
	\frac{\rho(x)}{N} = \rhor{\mu}(x) := \frac{1}{\pi} \frac{\mu}{\cosh(\mu x)} \label{eq:rho line}
\end{align}
and
\begin{align} 
\gamma = \sum_{j=-\frac{N-1}{2}}^{\frac{N-1}{2}} \ketbra{\psi_j}{\psi_j}, \quad \psi_j(x) := \frac{1}{\pi^{\frac12}} \frac{\mu ^{\frac12}}{\sqrt{\cosh(\mu x)}} \left(\frac{\ee^{-\frac{\mu}{2} x} + \ii \ee^{\frac{\mu}{2}x}}{\sqrt{2 \cosh(\mu x)}}\right)^{4j}, \label{eq:orbitals line}
\end{align}
with the convention that the sum runs over half-integers for $N$ even and integers for $N$ odd. 
Equivalently,
\begin{align}
    \gamma(x,y) = \frac{\mu}{\pi} \frac{1}{\sqrt{\cosh(\mu x)}} \frac{\sin(2N(\arctan(\ee^{\mu x}) - \arctan(\ee^{\mu y}))}{\sin(2 (\arctan(\ee^{\mu x}) - \arctan(\ee^{\mu y})))} \frac{1}{\sqrt{\cosh(\mu y)}}. \label{eq:gamma line}
\end{align}
\end{theorem}

The next result gives an analogous gradient-corrected LT inequality on the one-dimensional torus $\T = \R/\Z$. In this case, we also need the following shell correction with a pre-factor depending on the parity of the number of particles
\begin{align}
    \cS(\rho) = \left(\int_\T \frac{1}{\rho(x)} \dd x\right)^{-1}. \label{eq:shell correction}
\end{align}

\begin{theorem}[Sharp gradient-corrected Lieb--Thirring inequality on the torus]
\label{thm:LT_torus}
Let $N \in \N$, and $\gamma \in \mathfrak{S}_N(\T)$ be a one-particle density matrix on $L^2(\T)$, and define $\rho := \rho_\gamma$. Then
\begin{align}
    T[\gamma]+3G(\rho)
    \geq&
    \frac{\pi^2}{3}\left(1+\frac2{N^2}\right)
    \int_\T\rho^3
    -\mathbf{1}_{N\,\mathrm{odd}} \pi^2 \cS(\rho).
    \label{eq:LT_torus}
\end{align}
Moreover, equality is attained if and only if there exists some $m \in [0,1)$ such that up to translation,
\begin{align}
	\frac{\rho(x)}{N} = \rhot{m}(x) \coloneqq  \,\frac{2K(m)}{\pi}\,
    \sqrt{1-m\sin^2\left(\mathrm{am}\!\left(2K(m) x\,\middle|\,m\right)\right)},
    \label{eq:LT_torus_equality_density}
\end{align}
and
\begin{align} \gamma = \begin{dcases} \gamma_m, \quad &\mbox{if $N$ is odd}, \\ (1-t) \ee^{\ii\nu_m} \gamma_m \ee^{-\ii \nu_m}  + t \ee^{-\ii \nu_m} \gamma_m \ee^{\ii \nu_m} \quad t\in [0,1], \quad &\mbox{if $N$ is even}, \end{dcases} \label{eq:LT_torus_equality_gamma_cases}
\end{align}
where the phase factor is given by
\begin{align}
    \ee^{\ii\nu_m(x)}
    :=
    \frac{
    \cos\left(\mathrm{am}(2K(m)x\,|\,m)\right)
    +\ii\sqrt{1-m}\,
    \sin\left(\mathrm{am}(2K(m)x\,|\,m)\right)
    }{
    \sqrt{1-m\sin^2\left(
    \mathrm{am}(2K(m)x\,|\,m)\right)}
    } \label{eq:phase function}
\end{align}
and $\gamma_m$ is defined as
\begin{align} \gamma_m = \sum_{|k| \leq \frac{N-1}{2}} \ketbra{\psi_k}{\psi_k}, \quad \text{with}\quad \psi_k(x) := \sqrt{\rhot{m}(x)} \ee^{\ii 2 k \mathrm{am}(2 K(m) x |  m)}, \label{eq:LT_torus_equality_gamma}\end{align}
with the sum running over half-integers for $N$ even and integers for $N$ odd. The function $\mathrm{am}(\,\cdot\,|m)$ is the Jacobi amplitude, defined as the inverse of the incomplete elliptic integral
\begin{align}
    F(\varphi|m)
    :=
    \int_0^\varphi
    \frac{\dd\theta}{\sqrt{1-m\sin^2\theta}}, \label{eq:elliptic function def}
\end{align}
and $K(m) = F(\pi/2 \, | \, m)$ is the complete elliptic integral.
\end{theorem}
We now make several remarks on the results of Theorems~\ref{thm:LT_line} and~\ref{thm:LT_torus}.

\begin{remark}[General particle number]\label{rem:general N line}
For states $\gamma$ with a fractional number of particles $\Tr \gamma = N \in \R_+$ we have the following general versions of~\eqref{eq:LT_line_main} and~\eqref{eq:LT_torus}:
\begin{align}
    &T[\gamma] + 3G(\rho) \geq \frac{\pi^2}{3} \left(1+\frac{\max\{\kappa(N,0),\kappa(N,\pi)\}}{N^2}\right) \int_\R \rho^3, \label{eq:LT_general_line}
    \intertext{and}
    &T[\gamma] + 3G(\rho) \geq \frac{\pi^2}{3} \int_\T \rho^3 + \pi^2 \max\left\{ \frac{\kappa(N,\pi)}{3N^2} \int_\T \rho^3 - \cS(\rho), \frac{\kappa(N,0)}{3N^2} \int_\T \rho^3 \right\}, \label{eq:LT_general_torus}
\end{align}
where 
\begin{align*}
\kappa(N,\beta)
    \coloneqq
    \begin{cases}
    \displaystyle
    3 (\beta/\pi)^2 
    -6\frac{n(1-\theta)}{N}(\beta/\pi)
    +\frac{2n-3n\theta^2-\theta^3}{N},
    & n \text{ even},
    \\[3mm]
    \displaystyle
    3(\beta/\pi)^2 
    -6\frac{\theta(n+1)}{N}(\beta/\pi)
    +\frac{-n+6n\theta-3n\theta^2+3\theta-\theta^3}{N},
    & n \text{ odd}.
    \end{cases}
\end{align*}
with $n = \lfloor N \rfloor$ and $\theta = N-n$. Note that 
\begin{align} \max\{ \kappa(n,0), \kappa(n,\pi)\} = 2, \quad \text{and}\quad \min\{\kappa(n,0),\kappa(n,\pi)\} = -1 \quad\text{ for $n\in \N$,} \label{eq:simplified kappa} \end{align}
so the above inequalities reduce to~\eqref{eq:LT_line_main} and~\eqref{eq:LT_torus} for integer particle numbers. Moreover, in the non-integer cases $N\geq 2$, one can show that equality in~\eqref{eq:LT_general_line} is never attained while equality in~\eqref{eq:LT_general_torus} is achieved only for the homogeneous density $\rho = N$, see Lemma~\ref{lem:inequality saturation}.
\end{remark}

\begin{remark}[LDA rescaling]
For $\lambda>0$, set $\rho_\lambda(x):=\rho(\lambda x)$ and $N_\lambda:=N/\lambda$. Since $G(\rho_\lambda)=\lambda G(\rho)$ and $\int_\R\rho_\lambda^3=\lambda^{-1}\int_\R\rho^3$, for $\gamma_\lambda\in\mathfrak S_{N_\lambda}(\R)$ with density $\rho_\lambda$, inequality~\eqref{eq:LT_general_line} becomes
\begin{align}
    \lambda T[\gamma_\lambda]+3\lambda^2G(\rho)
    \geq \frac{\pi^2}{3} \left(1+\lambda^2 \frac{\max\{\kappa(N_\lambda,0), \kappa(N_\lambda,\pi)\}}{N^2}\right)\int_\R\rho^3.
    \label{eq:LT_line_LDA}
\end{align}
Thus both corrections are of order $\lambda^2$ in the LDA limit $\lambda\downarrow0$, and when $\lambda \to +\infty$,
\begin{align} \liminf_{\lambda \to + \infty} T[\gamma_\lambda]/\lambda +3 G(\rho)
    \geq \frac{\pi^2}{N^2}\int_\R\rho^3.
    \label{eq:LT_line_dilute_LDA}
\end{align}
\end{remark}

\begin{remark}[Optimal finite-rank constants] \label{rem:finite-rank} It follows from Theorem~\ref{thm:LT_line} that
\begin{align} c_N \coloneqq \frac{\pi^2}{3}\left(1+\frac{2}{N^2}\right)\label{eq:finite-rank cte} \end{align}
is the optimal finite-rank constant in front of the local term, i.e., 
\begin{align*}
    c_N = \min \left\{ \frac{T[\gamma] + 3G(\rho_\gamma)}{\int_\R \rho_\gamma^3} : 0\leq \gamma \leq 1, \quad 1\leq \mathrm{rank}(\gamma) \leq N \right\}, \quad N \in \N.
\end{align*}
A proof is provided in the appendix. Note that $c_N \downarrow c_{\rm TF}$ as $N\rightarrow \infty$.
\end{remark}

\begin{remark}[One-particle case] Inequality~\eqref{eq:LT_line_main} combined with the Hoffmann-Ostenhof bound \cite{HoffmannOstenhof77} $T[\gamma] \geq G(\rho)$ yields
\begin{align*}
    T[\gamma] \geq \frac{\pi^2}{4} \int_\R \rho^3, \quad \mbox{whenever $\int_\R \rho = 1$,}
\end{align*}
which is the sharp one-particle inequality on the real line \cite{SzNagy41,BenguriaLoss04}. On the torus we find instead
\[  T[\gamma] \geq \frac{\pi^2}{4} \int_\T \rho^3 - \frac{\pi^2}{4} \cS(\rho).\]
Note that the shell-term exactly compensates for the zero (or constant) mode $\rho = 1$. 
\end{remark}

\begin{remark}[From torus to real line optimizers] \label{rem:concentration profile} As $m\uparrow 1$, one has $K(m)\to+\infty$ and
\[
\rhot{m}(\cdot-x_0)
=
\frac{2K(m)}{\pi}
\sqrt{1-m\sin^2\left(\mathrm{am}\!\left(2K(m) (\cdot-x_0)\,\middle|\,m\right)\right)}
\rightharpoonup \delta_{x_0}
\]
weakly as probability measures on $\T$, see Figure~\ref{fig:torus optimizers}.  However, after introducing the blow-up scale $\varepsilon_m=(2K(m))^{-1}$, for any $y \in \R$,
\[
\varepsilon_m
\rhot{m}(\varepsilon_m y)
=
\frac1\pi\sqrt{1-m\sin^2\left(\mathrm{am}\!\left(y\,\middle|\,m\right)\right)}
\underset{\substack{m\uparrow 1}}{\longrightarrow} \; 
\frac{1}{\pi\cosh y}.
\]
Thus the concentration profile of the torus optimizers is
precisely the real-line optimizer.
\end{remark}

\begin{figure}[!ht]
\includegraphics[scale=0.3]{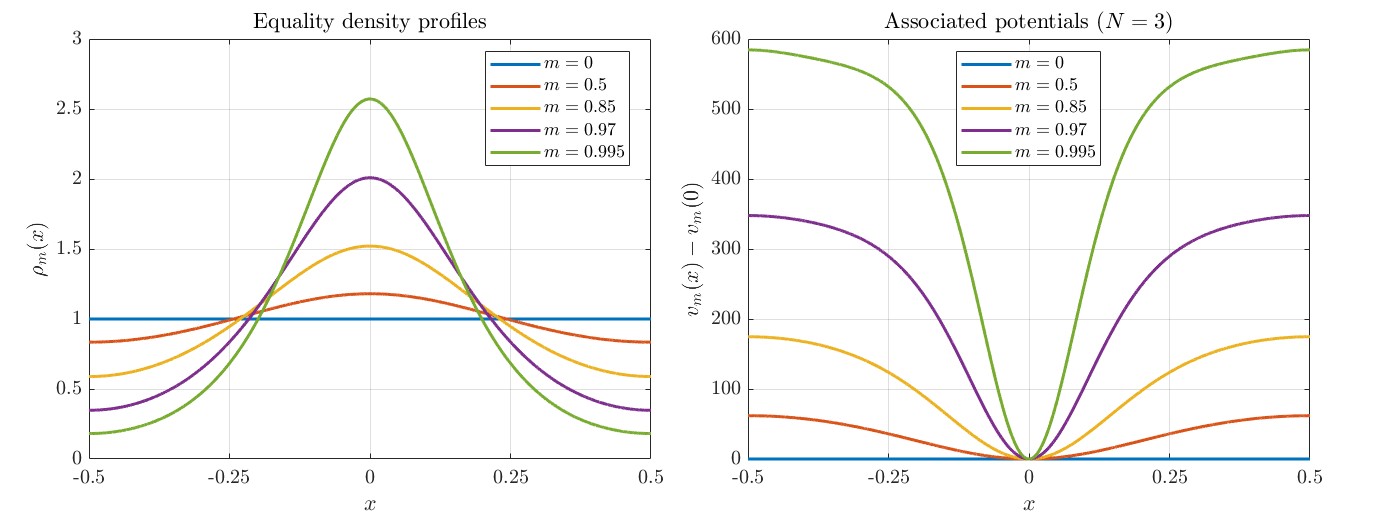}
\caption{
Plot of density profiles and corresponding Kohn-Sham potentials (for $N=3$) in the equality manifold in the torus case.
}
\label{fig:torus optimizers}
\end{figure}

\subsection{Applications to density functional theory}

We now present several applications to density functional theory. 

\subsubsection{Energy approximation}
As a first application, we provide a rigorous justification of the Thomas--Fermi-von Weizs\"acker approximation for the kinetic energy of one-dimensional systems. To state this result precisely, let us recall that the kinetic energy Lieb functional on $\Omega$ is defined as
\begin{align}
    T_{\rm L}^\Omega(\rho)
    :=
    \inf_{\substack{\gamma \in \mathfrak{S}(\Omega) \\ \rho_\gamma = \rho}}
    \Tr(-\partial_x^2\gamma).
    \label{eq:nonint_Lieb_functional} 
\end{align}
Moreover, let us define the Thomas--Fermi--von-Weizs\"acker kinetic-energy functional as
\begin{align}
    T_{\rm TFW}^\R(\rho) \coloneqq c_{\rm TF} \int_\R \rho^3 -\frac13 G(\rho),  \label{eq:TFW energy}
\end{align}
where $c_{\rm TF} = \pi^2/3$. Then, our first result is a bound on the error between $T_{\rm TFW}^\R$ and $T_{\rm L}^\R$ in terms of a Poincar\'e deficit for the normalized density profile.
\begin{theorem}[Thomas--Fermi approximation on $\R$ with gradient corrections] \label{thm:TFW_line} For any mass one density profile with $\sqrt{\rho} \in \mH^1(\R)$ and $N\in \N$, we have
\begin{align}
    - \frac{2}{3} N \delta_P^\R(\rho) \leq T_{\rm L}^\R(N\rho) - T_{\rm TFW}^\R(N \rho) \leq \frac13 N \delta_P^\R(\rho),
\end{align}
where $\delta_P$ is the cumulative Poincar\'e deficit defined as
\begin{align}
    \delta_P^\R(\rho) \coloneqq 4G(\rho) -\pi^2 \int_\R \rho^3. \label{eq:Poincare deficit real line}
\end{align}
In particular, $\delta_P^\R(\rho) \geq 0$ for any normalized density profile.
\end{theorem}

\begin{remark}[Dirichlet Poincar\'e deficit] \label{rem:Dirichlet Poincare deficit} As shown later in Section~\ref{sec:cumulative coordinates}, we can write
\begin{align*}
    \delta_P^\R(\rho) = \int_0^1 |a'(s)|^2 \dd s - \pi^2 \int_0^1 |a(s)|^2 \dd s,
\end{align*}
where $a(s)$ is the normalized density profile in cumulative coordinates. This is precisely the Poincar\'e deficit of $a$ for the Dirichlet Laplacian in the unit interval, which justifies the terminology used here.
\end{remark}

Our next result provides an analogous estimate in the periodic setting. In this case, besides a gradient correction, we need again a shell correction depending on the parity of the number of particles. Precisely, we define
\begin{align}
T_{\rm TFW}^{\T}(\rho) := c_{\rm TF} \int_\T \rho^3 -\frac13 G(\rho) + \theta_N \pi^2 \cS(\rho), \label{eq:TFWS def}
\end{align}
where $\cS(\rho)$ is defined in~\eqref{eq:shell correction} and $\theta_N$ is the parity factor
\begin{align}
    \theta_N := \begin{cases} -\frac13 , \quad &\mbox{whenever $N = \int_\T \rho$ is  odd}, \\
    \frac23, \quad &\mbox{whenever $N = \int_\T \rho$ is even.} \end{cases}
\end{align}
Then, with these definitions in place, we have the following result. 
\begin{theorem}[Thomas--Fermi approximation on $\T$ with gradient and shell corrections] \label{thm:TFW_torus} For any mass one density profile with $\sqrt{\rho}\in \mH^1(\T)$ and $N\in \N$, we have
\begin{align}
    -\frac23 N \delta_P^{\T}(\rho) \leq T_{\rm L}^{\T}(N\rho) - T_{\rm TFW}^{\T}(N \rho) \leq \frac13 N \delta_P^{\T}(\rho),
\end{align}
where $\delta_P^\T$ is the cumulative anti-periodic Poincar\'e deficit
\begin{align}
    \delta_P^\T(\rho) \coloneqq 4 G(\rho) + \pi^2 \cS(\rho) - \pi^2 \int_\T \rho^3. \label{eq:Poincare deficit torus}
\end{align}
In particular, $\delta_P^\T(\rho) \geq 0$ for any normalized density profile. 
\end{theorem}

\begin{remark}[Anti-periodic Poincar\'e deficit] \label{rem:anti-periodic Poincare deficit} As shown later in Section~\ref{sec:cumulative coordinates} we can write
\begin{align*}
    \delta_P^\T(\rho) = \int_0^1 |f_a'(s)|^2 \dd s - \pi^2 \int_0^1 |f_a(s)|^2 \dd s
\end{align*}
where $f_a(s) = a(s)\ee^{\ii \chi_a(s)}$ is the energetically optimal anti-periodization of the periodic density profile in cumulative coordinates. The right-hand side is precisely the Poincar\'e deficit of $f_a$ for the anti-periodic Laplacian in the unit interval.
\end{remark}

At first glance, it may seem strange to measure the energy error in terms of the Poincar\'e deficit. However, as shown later in Proposition~\ref{prop:zero deficit}, this deficit is zero if and only if the density comes from an optimizer of the gradient-corrected LT inequality. Consequently, Theorems~\ref{thm:TFW_line} and~\ref{thm:TFW_torus} give a rigorous justification of the Thomas--Fermi--von Weizs\"acker (plus shell correction) approximation whenever the normalized density profile is close to the two-dimensional Poincar\'e manifolds\footnote{Strictly speaking, $\mathcal{D}_P^\T$ is not a manifold around the homogeneous density $\rho_0$ as $\rho_0(x-x_0) = \rho_0$ for any $x_0\in\T$.}
\begin{align}
    &\mathcal{D}_P^{\R} \coloneqq \{ \rhor{\mu}(\cdot-x_0) : x_0 \in \R, \quad \mu >0, \quad \text{where $\rhor{\mu}$ is defined by~\eqref{eq:rho line}} \}, \label{eq:density manifold line}
    \intertext{and}
    &\mathcal{D}_P^\T \coloneqq \{ \rhot{m}(\cdot-x_0): x_0 \in \T, \quad m \in [0,1)  \quad\text{where $\rhot{m}$ is defined by~\eqref{eq:LT_torus_equality_density}} \}. \label{eq:density manifold torus}
\end{align}
In particular, these approximations are \emph{exact} on these manifolds.
\begin{corollary}[Exact evaluation on the optimizing density manifold]
\label{cor:exact_manifold} Let $\Omega = \R$ or $\Omega = \T$. Then, for any $N\in \N$,
\begin{align*}
    &T_{\rm L}^\Omega(N \rho) = T_{\rm TFW}^\Omega(N \rho), \quad \text{for any $\rho \in \mathcal{D}_P^\Omega$,}
\end{align*}
where $\mathcal D_P^\R$ and $\mathcal D_P^\T$ are the manifolds defined in~\eqref{eq:density manifold line} and~\eqref{eq:density manifold torus}. Note that the case $m=0$ in $\mathcal{D}_P^\T$ corresponds to the homogeneous density. 
\end{corollary}

\begin{remark}[Non-exact approximation away from the equality manifold?]
We do not know whether exactness can also occur for densities outside this family. Proving such a statement would be very interesting but possibly requires a careful stability analysis of the gradient-corrected LT inequality derived here. This is left for future work.
\end{remark}
\subsubsection{Local kinetic-energy density approximation}
The next result shows that the TFW approximation can be justified even at the level of energy densities. More precisely, for a given density $\sqrt{\rho} \in \mH^1(\Omega)$, where $\Omega = \T$ or $\Omega = \R$, we denote by $\gamma_\rho$ an optimizer of the kinetic energy Lieb functional, i.e.,
\begin{align*}
   \gamma_\rho \in \mathfrak S(\Omega), \qquad  T[\gamma_\rho] = T_{\rm L}^\Omega(\rho), \qquad\text{and}\qquad  \rho_{\gamma_\rho} = \rho. 
\end{align*}
We then define the local kinetic-energy density of $\gamma_\rho$ as 
\begin{align}
    \tau_\rho(x) := (-\partial_x \gamma_\rho \partial_x)(x,x) = \sum_{j \in M} |\psi_j'(x)|^2, \quad \text{where}\quad \gamma_\rho = \sum_{j \in M} \ketbra{\psi_j}{\psi_j},  \label{eq:exact kinetic-energy density}
\end{align}
and $M \subset \N$ is an index set. Since $\gamma_\rho$ is one of the optimizers, $\tau_\rho$ depends on the choice of the optimizer, but the notation does not specify it. For the TFW approximation, we define the local energy density as
\begin{align}
    \tau^{\rm TFW}_\rho := \begin{dcases} c_{\rm TF} \rho^3 + \frac13 \rho'' -  \frac{1}{12} \frac{|\rho'|^2}{\rho}, \quad &\text{on the real line $\Omega = \R$,} \\ c_{\rm TF} \rho^3 + \frac13 \rho'' -  \frac{1}{12} \frac{|\rho'|^2}{\rho} + \theta_N \pi^2 \cS(\rho)^2 \frac{1}{\rho}, &\quad \text{on the torus $\Omega = \T$,} \end{dcases}
    \label{eq:TFW energy density}
\end{align}
where $\rho''$ is understood in the distributional sense. Note that (formally at least) 
\begin{align*}
    T_{\rm L}^\Omega(\rho) = \int_\Omega \tau_\rho(x) 
    \dd x  \qquad\text{and}\qquad T_{\rm TFW}^\Omega(\rho) = \int_\Omega \tau_\rho^{\rm TFW}(x) \dd x.
\end{align*}
To measure the error between the exact kinetic-energy density and its TFW approximation, we introduce the following metric. For $\sqrt{\rho} \in \mH^1(\Omega)$, we define $\mathcal{Y}_\rho$ as the space of continuous functions $\phi$ on $U_\rho \coloneqq \{x\in \Omega : \rho(x) > 0 \}$ such that $\phi$ is weakly differentiable on $U_\rho$ and
\begin{align}
	\norm{\phi}_{\mathcal{Y}_\rho(\Omega)}^2 := \mysup{x \in U_\rho} |\phi(x)|^2 + \int_{U_\rho} \rho(x) |\phi'(x)|^2 \dd x <\infty.  \label{eq:Yrho norm}
\end{align}
We then denote the corresponding dual norm by
\begin{align}
    \norm{F}_{\mathcal{Y}_\rho^\ast(\Omega)} \coloneqq \sup_{0 \not\equiv\phi \in \mathcal{Y}_\rho} \frac{|\inner{F, \phi}|}{\norm{\phi}_{\mathcal{Y}_\rho}}. \label{eq:dual Yrho topology}
\end{align}
This norm may seem slightly exotic at first, but as we point out in Remark~\ref{rem:Yrho topology} below, it is natural for our purposes. Using this norm, we have the following result.

\begin{theorem}[Kinetic-energy density approximation]
\label{thm:local kinetic-energy density}
Let $\Omega=\T$ or $\Omega=\R$, and let $\sqrt{\rho} \in \mH^1(\Omega)$ be a mass one density profile. For $N\in \N$, let $\tau_{N\rho}$ and $\tau_{N\rho}^{\rm TFW}$ be respectively the kinetic-energy density defined in~\eqref{eq:exact kinetic-energy density} and its TFW approximation defined in~\eqref{eq:TFW energy density}. Then we have
\begin{align}
    \norm{\tau_{N\rho}-\tau_{N\rho}^{\rm TFW}}
    _{\mathcal{Y}_\rho^\ast(\Omega)}
    \leq
    N\left[
         C(\rho)
        \sqrt{\delta_{\rm P}(\rho)}
        +
        \delta_{\rm P}(\rho)
    \right],
    \label{eq:local_kinetic_Yrho_estimate}
\end{align}
where $\norm{\cdot}_{\mathcal{Y}_\rho^\ast}$ is defined in~\eqref{eq:dual Yrho topology}, $\delta_P(\rho)$ denote the appropriate Poincar\'e deficits in~\eqref{eq:Poincare deficit real line} and~\eqref{eq:Poincare deficit torus}, and
\begin{align} C(\rho) =
    \frac53 + \frac43 \sqrt{G(\rho)} + \frac23 \pi \sqrt{\cS(\rho)}, 
    \label{eq:Crho_local_kinetic}
\end{align}
with the convention that $\cS(\rho) = 0$ if $\Omega = \R$. 
\end{theorem}

\begin{remark}[Stability of approximate minimizers] \label{rem:approximate minimizers} In fact, Theorem~\ref{thm:local kinetic-energy density} does not require \(\gamma\) to be an exact minimizer of $T_{\rm L}$. The same conclusion holds for approximate minimizers: for any $0\leq \gamma \leq 1$ with density $N\rho$, and such that 
\[ T[\gamma] < T_{\rm TFW}(N\rho) + \left(c-\frac23\right)N \delta_P(\rho), \quad \mbox{for some $c>0$},\]
the associated kinetic-energy density $\tau_\gamma (x) = (-\partial_x \gamma \partial_x)(x,x)$ satisfies 
\begin{align}
    \norm{\tau_\gamma - \tau_{N\rho}^{\rm TFW}}_{\mathcal{Y}_\rho^\ast} \leq N\left(\left(C(\rho) - 1 + \sqrt{c}\right)\sqrt{\delta_P(\rho)} + c \delta_P(\rho) \right). \label{eq:general estimate}
\end{align}
In particular, for any sequence $\{ \gamma_N\}_{N\in \N}$ with $0\leq \gamma_N \leq 1$ and density $N\rho$ we have
\[ T[\gamma_N] \leq c_{\rm TF}\int_\Omega \rho^3 N^3 + \mathcal{O}_\rho(N), \qquad \text{implies}\qquad \tau_{\gamma_N} = c_{\rm TF} \rho^3 N^3 + \mathcal{O}_\rho(N). \]
\end{remark}

\begin{remark}[On the $\mathcal{Y}_\rho$ topology] \label{rem:Yrho topology} Let us make a few comments about the $\mathcal{Y}_\rho$ topology introduced above. 
\begin{enumerate}
\item (Topology on $v$-representable densities) In the case $\Omega = \T$, for any density satisfying $0< m \leq \rho \leq M$, we have the equivalence 
\[ \min\{1, m\} \norm{\phi}_{\mH^1} \leq \norm{\phi}_{\mathcal{Y}_\rho} \leq (C_{\rm Sob} + M) \norm{\phi}_{\mH^1}, \]
where $C_{\rm Sob}$ is the optimal constant in the Sobolev embedding $\mH^1(\T) \hookrightarrow \mL^\infty(\T)$. In particular, in the set of $v$-representable densities 
\begin{align}
    \mathcal{D}_1(\T) := \left\{ \rho :\sqrt{\rho} \in \mH^1(\T), \quad \int_\T \rho = 1,\quad \rho(x)>0, \quad \mbox{for any $x\in \T$} \right\},
\end{align}
the error can be measured in the density-independent $\mH^{-1}$ topology at the cost of multiplying it by the factor $(C_{\rm Sob} + \norm{\rho}_{\mL^\infty})$. 
\item (von-Weizs\"acker energy density) For $\sqrt{\rho} \in\mH^1$, the space $\mathcal{Y}_\rho^\ast$ is the natural ambient space for the term $\rho''$ in~\eqref{eq:TFW energy density}. Indeed, from Cauchy-Schwarz we have, for any $\psi \in \cY_\rho$,
\begin{align*}
    |\inner{\rho'', \psi}| = \left|\int_\Omega \rho'(x) \psi'(x) \dd x\right| \leq \left(\int_\Omega \frac{|\rho'|^2}{\rho}\right)^\frac12 \left( \int_\Omega \rho |\psi'|^2 \right)^{\frac12} \leq 2G(\rho)^{\frac12} \norm{\psi}_{\mathcal{Y}_\rho}. 
\end{align*}
\end{enumerate}
\end{remark}

\subsubsection{The inverse Kohn--Sham problem}

We now apply the previous results to study the inverse Kohn-Sham potential associated to a given density. As the existence of such a potential is not yet understood in the real line setting, we shall focus here on the periodic setting, whose existence and uniqueness has been recently settled \cite{SutPenRug24,Cor25,Cor26a}.

In this setting, it was shown in these works that for any normalized density profile $\rho \in \mathcal{D}_1(\T)$ and any $N\in\N$, there exists a unique (modulo additive constants) potential $v_{\rm KS}(N\rho) \in \mH^{-1}(\T;\R)$ such that $N\rho$ is a ground-state density of the operator 
\[ -\Delta_{\T^N} + \sum_{j=1}^N v_{\rm KS}(x_j) \quad \text{acting on} \quad \bigwedge^N \mL^2(\T).\]
On the other hand, one can define the inverse Thomas--Fermi--von Weizs\"acker potential with shell correction for any such density as 
\begin{align}
    v_{\rm TFW}(\rho)
    :=-3c_{\rm TF} \rho^2
    -\frac13\frac{(\sqrt\rho)''}{\sqrt\rho}
    -\theta_N\frac{\pi^2 \cS(\rho)^2}{\rho^2}.
    \label{eq:second_order_potential_main}
\end{align}
Then the next result estimates the difference between $v_{\rm TFW}$ and $v_{\rm KS}$ in the natural Sobolev space of potentials modulo additive constants $\mH^{-1}(\T)/\{1\}$.

\begin{theorem}[Kohn--Sham potential approximation]
\label{thm:KS_high_density}
Let $\rho\in\cD_1(\T)$ and suppose that $\rho \in \mH^2(\T)$. Then for any $N\in \N$, we have
\begin{align}
    \norm{N^{-2} v_{\rm KS}(N\rho)
        -N^{-2}v_{\rm TFW}(N\rho)}_{\mH^{-1}(\T)/\{1\}}
    \leq
    \frac{D_\rho}{N^2} 
    \left(\sqrt{\delta_{\rm P}(\rho)}
        +\delta_{\rm P}(\rho)
    \right),
    \label{eq:KS_corrected_deficit_main}
\end{align}
where
\begin{align} 
	D_\rho = 4 C(\rho)(1+2C_{\rm Sob})\nor{1/\rho}{\infty}^2\norm{\rho}_{\mH^2},  
\end{align}
and $C(\rho)$ is the constant defined in~\eqref{eq:Crho_local_kinetic} and $C_{\rm Sob}$ is the optimal constant in the Sobolev embedding $\mH^1(\T) \hookrightarrow \mL^\infty(\T)$. In particular, for any $\rho \in \mathcal{D}_1(\T) \cap \mH^2(\T)$,
\begin{align}
    \frac{1}{N^2} v_{\rm KS}(N\rho) = v_{\rm TF}(\rho) + \mathcal{O}_\rho(N^{-2}), \qquad  \label{eq:TF leading}
\end{align}
where $v_{\rm TF}(\rho) = -\pi^2 \rho^2$ is the inverse Thomas--Fermi potential. Moreover, the rate of convergence is optimal in general.
\end{theorem}

\begin{remark}[Complementary asymptotic regimes] Equation~\eqref{eq:TF leading} concerns the leading order of the non-interacting
density-to-potential map in the high-density regime $N\rightarrow \infty$. In the companion paper \cite{CG26}, we address the leading asymptotics in the interacting case. The strongly interacting and fixed particle number limit of the representing potential on the torus has been studied by the first author in~\cite{Cor26c},
where convergence to the Kantorovich potential
is established.
\end{remark}

As a corollary of Theorem~\ref{thm:KS_high_density}, we can explicitly compute the Kohn-Sham potentials at the periodic Poincar\'e density manifolds. In fact, we can also compute the Kohn-Sham potentials associated to the real line density manifold. Remarkably, these potentials come from two classes of (quasi-)exactly solvable models in one-dimensional quantum mechanics: the P\"oschl-Teller potential \cite{PoschlTeller33} in the real line and the associated Lam\'e potential \cite{KhareSukhatme99} in the torus. 
\begin{corollary}[Inverse potentials on the optimizing density manifold]
\label{cor:KS_Poincare_manifold} Let $\mu>0$ and $\rho(x)=\rhor{\mu}(x)$ where $\rhor{\mu}$ is defined in~\eqref{eq:rho line}. Then, up to additive constants,
\begin{align}
    v_{\rm KS}(N\rho)(x) = v_{\rm TFW}(N\rho)(x)
    =
    -\mu^2\left(N^2-\frac14\right)
    \frac{1}{\cosh^2\!\big(\mu x\big)}.
    \label{eq:PT_inverse_potential}
\end{align}
Likewise, let $\rho(x) = \rhot{m}(x)$ for some $m\in [0,1)$, where $\rhot{m}$ is defined in \eqref{eq:LT_torus_equality_density}, then, up to additive constants,
\begin{multline}
    v_{\rm KS}(N\rho)(x) = v_{\rm TFW}(N\rho)(x)  \\
    = 4K(m)^2 m\left[
        \left(N^2-\frac14\right) \mathrm{sn}^2\left(2K(m)x \, | \, m\right) + \left(\theta_N+\frac{1}{12}\right) \frac{\mathrm{cn}^2\left(2K(m) x \, | \, m\right)}{\mathrm{dn}^2\left(2K(m) x \, | \, m\right)}  \right],
    \label{eq:Lame_inverse_potential}
\end{multline}
where 
\[ \mathrm{sn}(x | m ) = \sin(\mathrm{am}(x | m)), \quad\! \mathrm{cn}(x | m ) = \cos(\mathrm{am}(x | m)), \quad \!\text{and}\!\quad\mathrm{dn}(x | m) = \sqrt{1-m \,\mathrm{sn}^2(x | m)} \]
are the Jacobi elliptic functions with modulus $\sqrt{m}$ and $\mathrm{am}$ is the Jacobi amplitude. 
\end{corollary}

\begin{remark}[Applications to numerical Kohn-Sham inversion] Corollary~\ref{cor:KS_Poincare_manifold} provides a set of benchmark densities and potentials for numerical inverse Kohn-Sham schemes \cite{JensenWasserman18,Garrigue22,HerbstBakkestuenLaestadius25}. Moreover, Theorem~\ref{thm:KS_high_density} provides computable error bounds for the inverse TFW potential approximation, which can be used as an initial guess for such schemes. In particular, equation~\eqref{eq:TF leading} proves the leading convergence observed numerically in~\cite[Section 5.5]{Garrigue22}. \end{remark}

\begin{remark}[TFW approximation estimates via stability of the density-to-potential map]
In~\cite{CorLae25}, it was shown that, for odd $N$, the density-to-potential map
\[
\rho\in\mathcal D_1^\T
\mapsto v_{\rm KS}(N\rho)
\in H^{-1}(\T;\R)/\R
\]
is locally Lipschitz continuous with respect to the $H^1(\T)$ topology. Combining this result with Corollary~\ref{cor:KS_Poincare_manifold} and the local Lipschitz continuity of the inverse TFW density-to-potential map, we obtain, for every $\rho_*\in\mathcal D_P^\T$, an estimate
\[
\|v_{\rm KS}(N\rho)-v_{\rm TFW}(N\rho)\|_{H^{-1}(\T)/\R}
\leq
C_{N,\rho_*}\,
\operatorname{dist}_{H^1}(\rho,\mathcal D_P^\T)
\]
for all $\rho\in\mathcal D_1^\T$ in a sufficiently small $H^1$-neighborhood of $\rho_*$. Note that this estimate is similar to~\eqref{eq:KS_corrected_deficit_main}, but requires only $H^1$ regularity of the density. Unfortunately, this estimate holds only locally, and the Lipschitz constant obtained in~\cite{CorLae25} relies on a compactness argument. Consequently, this approach does not provide explicit constants or quantitative dependence on $N$.
\end{remark}

\subsection{Outline of the proofs} Let us now briefly outline the main steps in the proofs of the main results and how they are distributed throughout the paper.

The proofs of Theorems~\ref{thm:LT_line} and~\ref{thm:LT_torus} are given in Section~\ref{sec:LT_gradient} and go roughly as follows. First, we pass to cumulative-mass coordinates, which transform any density matrix in $\R$ (or $\T$) with a prescribed density into a density matrix on the unit interval $I=(0,1)$ with pointwise constant density. This is the main step where the one-dimensional structure enters the proof. In the new coordinates, the desired inequalities reduce to lower bounds for weighted kinetic energies of density matrices subject to a constant-density constraint. This density constraint is crucial, for it allows us to use the spectral decomposition of the Laplacian with twisted boundary conditions and a Fourier-filling argument to obtain a general weighted kinetic energy inequality (see Lemma~\ref{lem:twisted_filling}), which is the central estimate of the paper. The strongest lower bound in this inequality is obtained with periodic boundary conditions for even $N$ and antiperiodic boundary conditions for odd $N$. On the line, the weighted orbitals vanish at both endpoints, so either twist is admissible and Theorem~\ref{thm:LT_line} follows at once from the weighted inequality and exact identities. On the torus, the antiperiodic choice requires an additional phase construction, which produces the shell correction. Finally, we characterize the equality cases of the weighted inequality in Lemma~\ref{lem:inequality saturation} and transform them back to physical coordinates to identify the optimizers in Theorems~\ref{thm:LT_line} and~\ref{thm:LT_torus}. This is the most calculation heavy step of the proof.

In Section~\ref{sec:TFW energy}, we combine the gradient-corrected Lieb--Thirring lower bounds with a matching upper bound to prove Theorems~\ref{thm:TFW_line} and~\ref{thm:TFW_torus}. For this upper bound, we use trial density matrices built from Harriman cumulative orbitals \cite{Harriman81} with centered periodic Fourier modes for odd $N$ and centered antiperiodic modes for even $N$, as in March and Young's construction \cite{MarchYoung58}. However, in the torus case, we need an additional common phase construction to restore periodicity at a minimal energy cost, which is crucial to obtain bounds controlled by the Poincar\'e deficit. In cumulative-mass coordinates, these trial states are inspired by the ones saturating the weighted inequality from Lemma~\ref{lem:twisted_filling}.

Theorem~\ref{thm:local kinetic-energy density} is proved in Section~\ref{sec:local kinetic energy}. We again work in cumulative-mass coordinates, where we can leverage the freedom to choose the weight in our central weighted inequality to obtain a local comparison inequality. Testing this inequality with weights of the form $\sqrt{1\pm tr}$ and using the global energy upper bound allows us to control the integral of the kinetic-energy density against arbitrary admissible test functions $r$. Optimizing in $t$ gives an error controlled by the square root of the cumulative Poincar\'e deficit and the deficit itself. To compare with the TFW kinetic-energy density, we also estimate a local Poincar\'e residual by the square root of the Poincar\'e deficit, see Lemma~\ref{lem:poincare residual bound}. These steps yield a density-constrained local Weyl estimate in cumulative coordinates (see Lemma~\ref{lem:weak_local_Fermi}), which is the key estimate in this proof. To complete the proof, we can then transfer this Weyl estimate back to physical coordinates. 

Finally, we prove Theorem~\ref{thm:KS_high_density} in Section~\ref{sec:KS potential}. This is done by combining the stationary force-balance equation, which relates potentials to kinetic energy densities, with the estimates in Theorem~\ref{thm:local kinetic-energy density}. Unfortunately, this argument requires the additional regularity assumption $\rho\in\mH^2(\T)$, because the force-balance equation formally contains products such as $v'\rho$, which are not generally well-defined for $v\in\mH^{-1}(\T)$ and $\rho\in\mH^1(\T)$. The additional derivative is also needed to recover estimates for the potential difference in $\mH^{-1}(\T)$, modulo additive constants. 

The proofs of Corollaries~\ref{cor:exact_manifold} and~\ref{cor:KS_Poincare_manifold} are presented, respectively, in Sections~\ref{sec:TFW energy} and~\ref{sec:KS potential}.

\subsection{Connection with previous works}

We conclude this section by briefly discussing some connections of the results obtained here with previous works in the literature.

\subsubsection*{Exactly solvable models and DFT approximations}

The inverse potentials associated with the Poincar\'e density manifolds belong to important families of exactly and quasi-exactly solvable one-dimensional Schr\"odinger operators. On the real line, they are given, up to translations and dilations, by the hyperbolic P\"oschl--Teller well, whose spectral theory is classical; see, e.g.,~\cite{PoschlTeller33,Lekner07}. On the torus, we are instead led to associated Lam\'e potentials of the form
\[
    m\left[
        a(a+1)\mathrm{sn}^2(u|m)
        +
        b(b+1)\frac{\mathrm{cn}^2(u|m)}{\mathrm{dn}^2(u|m)}
    \right],
    \qquad m\in[0,1),
\]
with the distinguished half-integer parameters
\[ a = N -\frac12, \quad b = \pm \frac 12\]
where the sign of $b$ depends on the parity of $N$. Associated Lam\'e potentials have been studied extensively in connection with algebraic solvability, finite-gap theory, and quasi-exact solvability, particularly at integer and half-integer values of the parameters; see, for example,~\cite{KhareSukhatme01,KhareSukhatme99,FernandezGanguly07}. However, we are not aware of any works in the literature that connect associated Lam\'e potentials to density functional approximations.

On the other hand, the P\"oschl--Teller well has been used as an exactly solvable test case for semiclassical expansions of kinetic-energy approximations in several works~\cite{Burke20,RibeiroBurke17,OkunCancioBurke24}. These works investigate the structure of semiclassical corrections to Thomas--Fermi theory in one dimension and emphasize the role of quantum oscillations and finite-system corrections beyond standard semilocal approximations in orbital-free DFT. In particular, Burke~\cite[eq.~(35)]{Burke20} derives an asymptotic correction arising from discrete orbital filling that closely resembles the shell correction obtained here. More generally, the decomposition into a smooth extended-Thomas--Fermi contribution and a finite-size shell correction has a long history in density-functional descriptions of finite systems, see, e.g., the work of Yannouleas and Landman~\cite{YannouleasLandman93,YannouleasLandman13}.

Nevertheless, the observation that the TFW approximation (with shell correction in the torus) is exact on the manifold of densities arising from a specific family of such exactly and quasi-exactly solvable potentials seems to be completely new. Remarkably, this exactness is true even at the level of energy densities. 

\subsubsection*{Lieb--Thirring inequalities}
 To the best of our knowledge, the first LT inequality with gradient terms was proved by Nam~\cite{Nam18}. In any dimension $d\in\N$, his estimate takes the form
\begin{align}
    T[\gamma]
    \geq c_{\rm TF,d}(1-\varepsilon)
    \int_{\R^d} \rho^{1+\frac2d}
    -
    \frac{c_d}{\varepsilon^{\beta_d}}\int_{\R^d} |\nabla \sqrt{\rho} |^2 , \label{eq:Nam est}
\end{align}
where $c_{\rm TF,d}$ denotes the $d$-dimensional semiclassical, or Thomas--Fermi, constant and $c_d,\beta_d>0$ depend only on the dimension. Seiringer and Solovej~\cite{SeiringerSolovej23} subsequently gave a simpler proof and obtained an improved dependence on $\varepsilon$. While these estimates hold in arbitrary dimensions, there is an
important distinction between~\eqref{eq:Nam est} and the inequalities derived here. In the former, the gradient term is seen as an error term: the coefficient of the local term approaches the Thomas--Fermi constant only at the expense of exploding the coefficient in front of the gradient term. In our one-dimensional inequality,
by contrast, the coefficient of the local term is strictly larger than the Thomas--Fermi constant, while the coefficient of the gradient term remains finite.

This distinction is particularly relevant in the low-dimensional setting. Indeed, it was recently shown by Read~\cite{Read26} that
\[
    T[\gamma] \geq \frac{\pi^2}{4}\int_\R \rho^3
\]
is the sharp one-dimensional Lieb--Thirring inequality. Since the local coefficient in~\eqref{eq:LT_line_main} exceeds the coefficient $\pi^2/4$ in the above inequality, the former cannot be derived by the latter alone. In this sense, the inequality derived here provides a genuine gradient correction, encoding information beyond the optimal local Lieb--Thirring inequality.

There is also a substantial literature on Lieb--Thirring inequalities in periodic settings. Work by Frank, Gontier, and Lewin~\cite{FrankGontierLewin21} developed a periodic formulation of the Lieb--Thirring problem and studied the connection with the optimal constants in whole-space settings. Interestingly, they proved that in one dimension, at the distinguished exponent $\gamma=3/2$, the spectral version of the periodic LT inequality admits a one-parameter family of Lam\'e optimizers. Although their family of Lam\'e potentials differs from the associated Lam\'e potentials identified here, the appearance of these quasi-exactly solvable potentials in both settings is rather intriguing. Inequalities on tori and related compact manifolds, including estimates of the corresponding constants, were also studied by Ilyin, Laptev, Zelik, and others~\cite{ZelikIlyinLaptev19,IlyinLaptevZelik20}. These inequalities have applications to estimating the dimension of global attractors for two-dimensional Navier--Stokes systems and typically impose zero-mean conditions on the orbitals. Our results complement this literature by incorporating an explicit shell correction alongside the local and gradient contributions. 

The dependence of Lieb--Thirring inequalities on the particle number has also been studied, for instance, in~\cite{FrankGontierLewin21,FrankGontierLewinFiniteRank}. These works investigate the existence of optimizers and how the optimal constant changes when the number of eigenvalues, or equivalently, the rank of the density matrix, is restricted. Here, as highlighted in Remark~\ref{rem:finite-rank}, we address the analogous finite-rank question in the one-dimensional gradient-corrected case (with our prescribed choice of pre-factor for the gradient term). 

\section{Gradient-corrected Lieb--Thirring inequalities and their optimizers}
\label{sec:LT_gradient}

In this section our goal is to prove Theorems~\ref{thm:LT_line} and~\ref{thm:LT_torus}.
\subsection{Cumulative-mass coordinate and exact identities}
\label{sec:cumulative coordinates}

We begin by introducing the cumulative-mass coordinates. This transformation is a key feature of the one-dimensional setting and will play an important role throughout the paper.  

Let $\rho$ be a normalized density function, i.e., satisfying $\rho \geq 0$ and $\int_\Omega \rho = 1$. Then we define the cumulative coordinate as
\begin{align}
    s = S(x) \coloneqq \int_{x_0}^x\rho(y)\,\dd y,
    \label{eq:S def}
\end{align}
where $x_0=-\infty$ on $\R$ and $x_0\cong 0$ on $\T$. Moreover, we define its left-continuous inverse as
\begin{align}
    X(s) = \sup \{ y \in \Omega : S(y) < s\}, \qquad s\in(0,1).
\end{align}
In the torus case, we identify $\Omega$ with the ordered interval $[0,1)$ when defining $S$ and taking this supremum.
We can now introduce the cumulative density as
\begin{align}
    \qquad
    a(s):=\rho(X(s)). \label{eq:a def}
\end{align}
In the next proposition, we summarize some elementary but important properties of the cumulative coordinates and cumulative density.

\begin{proposition}[Structural properties of cumulative coordinates and density] \label{prop:structure cumulative density} Let $\Omega = \R$ or $\Omega = \T$ and $\sqrt{\rho} \in \mH^1(\Omega)$ be a mass one density. Let $S,X$ and $a$ be the cumulative coordinates and density defined in~\eqref{eq:S def}--\eqref{eq:a def}. Then the following holds: 
\begin{enumerate}[label=(\roman*)]
\item (Component-wise diffeomorphism) \label{it:structure cumulative coordinates} We have the decomposition with $Q \subset \N$, 
\[ \{s \in (0,1): a(s)> 0 \} = \cup_{\alpha \in Q} S_\alpha \]
where each $S_\alpha$ is an open interval such that $X: S_\alpha \rightarrow J_\alpha \coloneqq X(S_\alpha)$ is a $C^1$-diffeomorphism with $X'(s) = 1/a(s)$. Moreover, 
\begin{align*}
    \cup_{\alpha \in Q} J_\alpha = \{x \in \Omega\setminus \{x_0\} : \rho(x)>0 \}, \quad \text{and}\quad \sum_{\alpha \in Q} |S_\alpha| = 1.  
\end{align*}
\item ($\mH^1$-regularity and boundary conditions) \label{it:H1 regularity} $a$ belongs to $\mH^1(0,1)$ and 
\begin{align}
    a'(s) = \frac{\rho'(X(s))}{a(s)}, \quad \mbox{for a.e. $s \in (0,1)$.} \label{eq:a gradient}
\end{align}
Moreover, if $\Omega = \T$, respectively, $\Omega=\R$, then $a\in \mH^1(\T)$, respectively $a\in \mH^1_0(0,1)$. 
\item (Local and gradient terms) We have the identities
\begin{align}
    \int_0^1 |a(s)|^2 \dd s = \int_\Omega \rho(x)^3 \dd x \qquad \text{and}\quad \int_0^1 |a'(s)|^2 \dd s = 4 G(\rho). \label{eq:local and gradient terms} 
\end{align}
\end{enumerate}
\end{proposition}

\begin{proof} Since $\rho$ is continuous, the set $\{ x\in \Omega\setminus \{x_0\}: \rho>0\}$ is open and can be written as the union of its (countably many) connected components $J_\alpha$. Now, on each $J_\alpha$, the function $S:J_\alpha\rightarrow S(J_\alpha) =: S_\alpha$ is differentiable with strictly positive derivative, so the inverse $X\rvert_{S_\alpha} = (S\rvert_{J_\alpha})^{-1}$ is also strictly monotone, differentiable and satisfies $X'(s) = 1/S'(X(s)) = 1/a(s)$. Moreover, as $S(J_\alpha)$ are disjoint and $\rho$ integrates to $1$, we have
\begin{align*}
    \sum_{\alpha \in Q} |S_\alpha| = \sum_{\alpha \in Q} \int_{J_\alpha} \rho = \int_\Omega \rho = 1.
\end{align*}
Also by construction, $a>0$ on $\cup_{\alpha\in Q}S_\alpha$. For the converse, fix $s\in(0,1)$. The set $\{y:S(y)<s\}$ is nonempty and bounded above, so its supremum $X(s)$ is finite. By continuity of $S$, we have $S(X(s))\leq s$. If this inequality were strict, continuity would give a point $y>X(s)$ with $S(y)<s$, contradicting the definition of $X(s)$. Thus $S(X(s))=s$. If $a(s)=\rho(X(s))>0$, then $X(s)\in J_\alpha$ for some $\alpha\in Q$, and consequently $s=S(X(s))\in S_\alpha$. This proves~\ref{it:structure cumulative coordinates}. At a plateau value in $(0,1)$, the left-continuous inverse selects the left endpoint of the plateau, where $\rho=0$, consistently with this characterization of the positivity set.

To prove~\ref{it:H1 regularity}, we first note that, for any $s\in S_\alpha$ converging towards the end points of $S_\alpha$, $X(s)$ must also converge to the corresponding end point of $J_\alpha$. Hence, at any endpoint of $S_\alpha$ inside the interval $(0,1)$, by continuity, we have $a(s) \rightarrow 0$ as $s$ approaches such endpoints. On the other hand, inside the interval $S_\alpha$ we have
\begin{align*}
    a'(s) = \rho'(X(s))/\rho(X(s))
\end{align*}
by the usual chain rule and
\begin{align}
    \int_{S_\alpha} |a|^2 = \int_{J_\alpha} \rho^3 \quad \text{and}\quad \int_{S_\alpha} |a'|^2 = 4\int_{J_\alpha} |(\rho^{\frac12})'|^2 \label{eq:gradient and local componentwise}
\end{align}
by the classical change of variables formula. Therefore, $a\in \mH^1(S_\alpha)$ has zero trace at any endpoint inside $(0,1)$. In particular, we can extend it by zero to the rest of $(0,1)$. Thus, since $G(\rho)<\infty$, we can enumerate the intervals $S_\alpha$ and take the $\mH^1$-limit in the sum $\sum_{\alpha \in Q} a\rvert_{S_\alpha}$ to conclude that $a = \sum_{\alpha \in Q} a\rvert_{S_\alpha} \in \mH^1(0,1)$. 

For the boundary value statements, note that we already showed that each $a\rvert_{S_\alpha}$ has zero boundary trace if $0,1\not\in \overline{S_\alpha}$. Hence, if no interval has endpoints $0$ or $1$, we have $a\in \mH^1_0(0,1)$ because $a$ is the $\mH^1$-limit of $\mH^1_0$ functions. On the other hand, if we have an interval with $0$ or $1$ as its endpoint, then the value of $a$ at this point can be evaluated by taking the limit of $\rho(x)$ as $x$ approaches the corresponding endpoint of $J_\alpha$. Consequently, the result on $\R$ follows from the fact that any $f\in \mH^1(\R)$ vanishes at infinity (as it can be uniformly approximated by compactly supported functions) and the result on $\T$ follows by periodicity. 

Finally,~\eqref{eq:local and gradient terms} follows by summing~\eqref{eq:gradient and local componentwise}. 
\end{proof}

Next, let us denote by $U_\rho: \mL^2(0,1) \rightarrow \mL^2(\Omega)$, the linear map
\begin{align}
    (U_\rho f)(x) = f(S(x)) \sqrt{\rho(x)}, \quad \text{where $S$ is the cumulative function~\eqref{eq:S def}.} \label{eq:cumulative isometry def}
\end{align}
It is immediate from the change of variables formula that $U_\rho$ is an isometry. In particular, we can pull-back a density matrix in $\Omega$ to a density matrix in $(0,1)$ via the formula
\begin{align} Q = U_\rho^\ast \gamma U_\rho. \label{eq:density matrix pull-back} \end{align}
The next observation, which is crucial to our analysis, is that one can relate the kinetic energy of $\gamma$ to a weighted kinetic energy of $Q$.
\begin{lemma}[Exact kinetic energy identity in cumulative-coordinates]
\label{lem:cumulative_kinetic_identity}
Let $\sqrt{\rho}\in\mH^1(\Omega)$ be a mass one density on $\Omega = \R$ or $\Omega = \T$, and let $a$ be defined as in~\eqref{eq:a def}. Let $0\leq \gamma \leq 1$ be a density matrix with density $N \rho$ for some $N>0$. Let $Q$ be the pull-back defined via~\eqref{eq:density matrix pull-back}, then the following holds:
\begin{enumerate}[label=(\roman*)]
\item (Orthonormality and density constraint) \label{it:constant density constraint} We have
\begin{align}
    0\leq Q\leq1,
    \qquad \mbox{and}\quad 
    Q(s,s)= N.
    \label{eq:Q_constraints}
\end{align}
\item (Kinetic energy identity) \label{it:kinetic identity} If $T[\gamma] < \infty$, then 
\begin{align}
\Tr -\partial_s (a Q a)\partial_s = 
    T[\gamma]+3N G(\rho).  \label{eq:kinetic identity}
\end{align}
\end{enumerate}
\end{lemma}
\begin{proof} The bound $0\leq Q \leq 1$ is immediate. For the density identity, let $P:=\mathbf1_{\{\rho>0\}}$. The density constraint gives $\gamma=P\gamma P$, since occupied orbitals vanish a.e. on $\{\rho=0\}$. Using $U_\rho fU_\rho^\ast=(f\circ S)P$, for any $f\in C([0,1])$ we obtain
\[ \Tr Q f = \Tr \gamma U_\rho f U_\rho^\ast = \Tr \gamma  f \circ S = \int_\Omega N \rho(x) f(S(x)) \dd x = N \int_0^1 f(s) \dd s. \]
Hence, the density of $Q$ is $N$ a.e., proving~\ref{it:constant density constraint}. 

To prove~\ref{it:kinetic identity}, write the spectral decomposition
\[
    \gamma=\sum_j n_j\ketbra{u_j}{u_j},
    \qquad 0<n_j\leq1, \qquad \sum_j n_j=N,
\]
where the occupied eigenfunctions $u_j$ are orthonormal. Set $\psi_j:=\sqrt{n_j}u_j$, so that $\gamma=\sum_j\ketbra{\psi_j}{\psi_j}$ and $\sum_j|\psi_j|^2=N\rho$. Since $T[\gamma]<\infty$, each $u_j$ belongs to $\mH^1(\Omega)$ and $T[\gamma]=\sum_j n_j\int_\Omega|u_j'|^2=\sum_j\int_\Omega|\psi_j'|^2$. Then $Q=\sum_j\ketbra{\phi_j}{\phi_j}$, where $\phi_j=U_\rho^\ast\psi_j$ are the transformed orbitals
\begin{align}
    \phi_j(s) := \frac{\psi_j(X(s))}{a(s)^{\frac12}}. 
\end{align}
So, assuming first that $\rho>0$ everywhere in $\Omega$ so that $X = S^{-1}$ is a $C^1$ diffeomorphism with $X'(s) = 1/a(s)>0$, the functions $\phi_j$ are weakly differentiable in $(0,1)$ and
\[ (a\phi_j)' = \frac{a'}{2a^\frac12} \psi_j\circ X + \frac{1}{a^{\frac12}} \partial_x \psi_j \circ X. \]
Thus, taking the absolute value square, summing over $j$, and using that
\[ \sum_{j} |\psi_j|^2 \circ X = N a \quad \mbox{and}\quad 2\sum_{j} \mathrm{Re} \, \overline \psi_j \partial_x \psi_j \circ X = N \partial_x \rho \circ X = N a' a,\]
we find that
\[ \sum_{j} |(a\phi_j)'|^2  = \frac34 N |a'|^2 + \frac{1}{a}\sum_{j} |\partial_x \psi_j|^2 \circ X. \]
Hence, integrating over $s$ and using the change of variables formula we obtain the desired identity. 

In the general case where $\rho$ vanishes at some point, one can use property~\ref{it:structure cumulative coordinates} and do the same calculations component-wise on the set $\{ \rho >0\}$. 
\end{proof}

\subsection{A twisted Fourier filling inequality}
The previous identities reduce the LT inequality to a lower bound of a weighted kinetic energy for density matrices with a constant density. This allows us to use a weighted inequality based on the explicit decomposition of the Laplacian with twisted boundary conditions. 

Precisely, for $\beta\in(-\pi,\pi]$, let $H_\beta = -\partial_s^2$ denote the self-adjoint realization of the Laplacian in $[0,1]$ with twisted form domain
\[
    \mathcal{Q}_\beta \coloneqq \{ f\in \mH^1(I) : \quad 
    f(1)=e^{i\beta}f(0) \}.
\]
We then recall the following well-known spectral decomposition of $H_\beta$. 
\begin{lemma}[Spectral decomposition of $H_\beta$] \label{lem:explicit Laplacian} The eigenmodes of $H_\beta$ are given by
\begin{align} \ee_j^\beta(x) \coloneqq \ee^{\rm i (2\pi j +\beta)x}, \quad \mbox{and}\quad \mu_j(\beta) = (2\pi j + \beta)^2, \qquad j \in \Z. \label{eq:eigenpairs explicit} \end{align}
In particular, denoting the eigenvalues with multiplicity in increasing order as
\[ \varepsilon_1(\beta) \leq \varepsilon_2(\beta) \leq ... \rightarrow \infty, \]
then, for $0\leq \beta \leq \pi$, we have
\begin{align}
\varepsilon_N(\pm \beta) = \begin{cases} (\pi N -\pi+\beta)^2,  \quad &\mbox{for $N$ odd,} \\
(\pi N - \beta)^2 ,\quad &\mbox{for $N$ even,} \end{cases}
\end{align}
and
\begin{align}
    \Sigma_N(\pm \beta) \coloneqq \sum_{j=1}^N \varepsilon_j(\beta) &=
    \begin{cases}
        \displaystyle\frac{\pi^2}{3}N(N^2-1) + \beta^2 N, \quad &\text{for $N$ odd},\\[2mm]
        \displaystyle\frac{\pi^2}{3}N(N^2-1) + (\pi - \beta)^2 N, \quad &\text{for $N$ even.}
    \end{cases} \label{eq:Fermi-Sea sum}
\end{align}
\end{lemma}
\begin{proof} It is immediate to verify that $\{\ee_j^\beta\}$ form an orthonormal sequence of eigenfunctions with corresponding eigenvalues given by $\{\mu_j\}$. Completeness follows from completeness of the usual Fourier basis $\{\ee_j^0\}_{j\in \Z}$. The explicit formula for the ordered eigenvalues and for the sum follow from the explicit formulas and straightforward calculations. 
\end{proof}

Using the above spectral decomposition, we can now prove the following weighted estimate, which is the only estimate used in our proof. This estimate will also play a central role in the local kinetic-energy density estimates.

\begin{lemma}[Twisted Fourier filling]
\label{lem:twisted_filling}
Consider $I=(0,1)$, $N>0$, and $\beta\in(-\pi,\pi]$. Let $Q \in \mathcal{B}(\mL^2(I))$ with $0\leq Q \leq 1$ and $\Tr Q = N$, decomposed in $Q = \sum_{j\in M} \ketbra{\phi_j}{\phi_j}$, where $M\subset\N$ and $\phi_j \in L^2(I)$ for every $j\in M$, and suppose that $Q(s,s) = N$. Let $w$ be a bounded function and suppose that $w \phi_j$ belongs to $\mathcal{Q}_\beta$ for every $j\in M$. Then
\begin{align}
    \sum_{j\in M}\int_0^1|(w \phi_j)'|^2
    \geq \left(\Sigma_{\lceil N \rceil}(\beta)+ \varepsilon_{\lceil N \rceil}(\beta)(N - \lceil N \rceil)\right) \int_0^1 |w|^2 .
    \label{eq:weighted_twisted_filling}
\end{align}
In particular,
\begin{align}
    \sum_{j\in M} \int_0^1 |(w\phi_j)'|^2
    \geq
    \frac{\pi^2}{3}
    \left(N^2+\kappa(N,\beta)\right)N
    \int_0^1 |w|^2 , \label{eq:simplified lower bound}
\end{align}
where, writing $m\coloneqq \lfloor N\rfloor$ and $\theta\coloneqq N-m\in[0,1)$, 
\begin{align}
    \kappa(N,\beta)
    \coloneqq
    \begin{cases}
    \displaystyle
    3 (\beta/\pi)^2 
    -6\frac{m(1-\theta)}{N}(|\beta|/\pi)
    +\frac{2m-3m\theta^2-\theta^3}{N},
    & m \text{ even},
    \\[3mm]
    \displaystyle
    3(\beta/\pi)^2 
    -6\frac{\theta(m+1)}{N}(|\beta|/\pi)
    +\frac{-m+6m\theta-3m\theta^2+3\theta-\theta^3}{N},
    & m \text{ odd}.
    \end{cases}
    \label{eq:kappa_N_beta}
\end{align}
\end{lemma}
\begin{proof} Let us set
\[  \lambda_F = \varepsilon_{\lceil N \rceil}(\beta)\quad \text{and}\quad \Gamma_w = w Q \overline w. \]
By the decomposition of $Q$, we have
\begin{align}
    \Tr H_\beta\Gamma_w = \sum_{j\in M}\int_0^1 |(w\phi_j)'|^2. \label{eq:trace orbital identity}
\end{align}
Then, since $Q(s,s) = N$, we have
\begin{align}
    \Tr \lambda_F \Gamma_w = N \varepsilon_{\lceil N \rceil}(\beta)\int_0^1 |w|^2. \label{eq:first trace id}
\end{align}
Similarly, since the eigenfunctions of $H_\beta$ have constant modulus one, we have
\begin{align}
    \Tr (H_{\beta} - \lambda_F)_- |w|^2 = \left(\varepsilon_{\lceil N \rceil}(\beta)\lceil N \rceil -\Sigma_{\lceil N \rceil}(\beta)\right) \int_0^1 |w|^2. \label{eq:second trace id}
\end{align}
Combining~\eqref{eq:first trace id} and~\eqref{eq:second trace id} we find
\begin{multline}
    \Tr H_{\beta} \Gamma_w - \left(\Sigma_{\lceil N \rceil}(\beta)+ \varepsilon_{\lceil N \rceil}(\beta)(N - \lceil N \rceil)\right) \int_0^1 |w|^2 = \Tr (H_{\beta}- \lambda_F) \Gamma_w + \Tr(H_{\beta} - \lambda_F)_- |w|^2  \ \\
    = \Tr  (H_{\beta} - \lambda_F)_+ \Gamma_w + \Tr (H_{\beta} - \lambda_F)_-(|w|^2 - \Gamma_w). \label{eq:funny id}
\end{multline}
To conclude, note that, since $0\leq Q \leq1$ we have $0\leq \Gamma_w \leq |w|^2$, and therefore, the right-hand side is the sum of two non-negative terms.

The explicit formulas in~\eqref{eq:simplified lower bound} and~\eqref{eq:kappa_N_beta} follow from the explicit formulas in Lemma~\ref{lem:explicit Laplacian}.
\end{proof}

We next explicitly compute all weight functions and density matrices that saturate inequality~\eqref{eq:weighted_twisted_filling}.

\begin{lemma}[Saturation of weighted inequality] \label{lem:inequality saturation} Let $N>0$, $\beta\in(-\pi,\pi]$, and let $w\geq0$ be a bounded
weight that is not identically zero. Let
\[
    Q=\sum_{j \in M}|\phi_j\rangle\langle\phi_j|,
    \qquad 0\leq Q\leq1,
    \qquad Q(s,s)=N, 
    \qquad M \subset \N
\]
and define
\[ \lambda_F \coloneqq \varepsilon_{\lceil N \rceil}(\beta). \]
Suppose that $w\phi_j\in\mathcal Q_\beta$ for every $j \in M$, then
\begin{align}
    \Tr H_{\beta} (w Q w) = \left(\Sigma_{\lceil N \rceil}(\beta)+ \lambda_F(N - \lceil N \rceil)\right) \int_0^1 w^2.  \label{eq:weighted_saturated}
\end{align}
holds if and only if one of the following holds:
\begin{enumerate}[label=(\roman*)]
\item \textbf{Integer exceptional case.} \label{it:exceptional}
If $N\geq2$ is an integer and $N+\beta/\pi$ is even, then
\begin{align}
    &w^2(s) = A + B \cos(2\pi (s-s_0)), \qquad A>0, \quad 0\leq B \leq A, \quad s_0 \in [0,1) \label{eq:weight_formula}
\end{align}
and
\begin{align}
&Q = t Q_+ + (1-t) Q_-, \quad t\in [0,1], \qquad
    Q_{\pm} = \frac{f_\pm}{w} P^{\leq}_{\pi-\beta} \frac{\overline{f_\pm}}{w} , \label{eq:Q formula}
\end{align}
where
\begin{align}
    f_+(s) \coloneqq p \ee^{-\ii \pi (s-s_0)} + q \ee^{\ii \pi (s-s_0)}, \qquad f_-(s) \coloneqq \overline{f_+(s)}, \label{eq:f formula}
\end{align}
with $p \geq q \geq 0$ satisfying
\begin{align}
    p^2 + q^2 = A, \quad 2p q = B, \label{eq:alpha beta def}
\end{align}
and $P^{\leq}_{\pi-\beta}$ denotes the (centered) rank-$N$ Fermi projection with the opposite boundary twist:
\[
P^{\leq}_{\pi-\beta} \coloneqq \mathrm{1}_{(-\infty,\lambda_F]}(H_{\pi-\beta}) = 
\sum_{\substack{k\in\frac{\pi -\beta}{2\pi}+\mathbb Z\\
|k|\leq (N-1)/2}}
\ketbra{\ee_k}{\ee_k},
\qquad
\ee_k(s):=\ee^{2\pi\ii ks}.
\]

\item
\textbf{Antiperiodic lowest-level case.} \label{it:lowest}
If $\beta=\pi$ and $0<N<2$, then $w^2$ is as in~\eqref{eq:weight_formula} and
\begin{equation}
    Q=N\bigl(tQ_++(1-t)Q_-\bigr),
    \qquad
    Q_\pm=
    \left|\frac{f_\pm}{w}\right\rangle
    \left\langle\frac{f_\pm}{w}\right|,
    \label{eq:saturation_lowest}
\end{equation}
where $f_{\pm}$ are again defined by~\eqref{eq:f formula} and~\eqref{eq:alpha beta def}, and $0\leq t\leq 1$ satisfies the additional restriction 
\begin{equation}
    2N^2 t(1-t) \frac{\sqrt{A^2- B^2}}{A + \sqrt{A^2-B^2}} \geq (N-1)_+.
    \label{eq:t restriction}
\end{equation}
In particular, every $t\in[0,1]$ is admissible when
$0<N\leq1$.
\item \textbf{Constant weights.}
If neither case~\ref{it:exceptional} nor case~\ref{it:lowest} apply, then $w^2=A$ for
some $A>0$, and
\begin{equation}
    Q=P^<_\beta +
    \sum_{k\in K_F}t_k|\ee_k\rangle\langle \ee_k|,
    \qquad
    0\leq t_k\leq1,
    \qquad
    \sum_{k\in K_F}t_k=N-\rank P^<_\beta,
    \label{eq:saturation_constant}
\end{equation}
where $P^<_\beta \coloneqq \mathrm{1}_{(-\infty,\lambda_F)}(H_\beta)$ and 
\[
    K_F:=
    \left\{
        k\in\mathbb Z+\frac{\beta}{2\pi}:
        (2\pi k)^2=\lambda_F
    \right\}.
\]
In particular, if $\lambda_F$ is simple, $Q$ is uniquely
determined.
\end{enumerate}
\end{lemma}

\begin{proof} As in the previous proof, we set $\Gamma_w 
\coloneqq w Q w$. Then, since both terms in~\eqref{eq:funny id} are non-negative, identity~\eqref{eq:weighted_saturated} holds if and only if both 
\begin{align} \Tr \left(H_{\beta} - \lambda_F\right)_+ \Gamma_w = 0 \quad \text{and}\quad \Tr \left(H_\beta - \lambda_F \right)_- (w^2 - \Gamma_w) = 0. \label{eq:vanishing two} \end{align}

\textbf{Step 1: constant weights.} Let us now assume that $w$ is constant and show that $Q$ is of the form appearing in~\eqref{eq:saturation_constant}. For this, and for the rest of the proof, it is convenient to introduce the subspaces
\begin{align*} &E^{<} \coloneqq \ran \mathrm{1}_{(-\infty, \lambda_F)}(H_\beta),  \quad\!\! E \coloneqq \ran \mathrm{1}_{\{\lambda_F\}}(H_\beta), \!\!\quad \text{and}\!\! \quad E^{\leq} \coloneqq E^{<} \oplus E. \end{align*}
Note that, by the definition of $\lambda_F$ and the spectral decomposition of $H_\beta$ we have
\begin{align}
    \dim E^{<} \leq \lceil N \rceil -1 \quad \text{and}\quad \dim E \leq 2.
\end{align}
Then, for constant $w$,~\eqref{eq:vanishing two} is equivalent to 
\[ \Tr  (H_\beta - \lambda_F)_+ Q = 0 = \Tr (H_\beta - \lambda_F)_- (1-Q). \]
Since $0\leq Q \leq 1$, this implies that $Q$ acts only on $E^{\leq}$ and satisfies $Q\rvert_{E^{<}} = 1$. In particular, $Q = P^<_\beta + H$ for some $H$ acting on $E$. Therefore, if $\dim E = 1$, by the trace constraint $\Tr Q = N$, we have
\[ Q = P^<_\beta + (N - \dim E^{<}) P_E \]
which is precisely~\eqref{eq:saturation_constant} for the case of simple eigenvalue $\lambda_F$. If $\dim E = 2$, then, since $P^<_\beta(s,s)$ is pointwise constant, the condition $Q(s,s) = N$ forces the off-diagonal part of $H$ (in the Fourier eigenbasis of $E$) to vanish. This implies that $Q$ is again of the form~\eqref{eq:saturation_constant}. We have therefore established that, if the weight is constant, then $Q$ is given by~\eqref{eq:saturation_constant}.

\textbf{Step 2: determining the weights.} We now proceed to characterize the weight functions. First, vanishing of the first term in~\eqref{eq:vanishing two} is equivalent to 
\begin{align} \ran \Gamma_w \subset E^{\leq}.\label{eq:first vanishing} \end{align}
On the other hand, vanishing of the second term in~\eqref{eq:vanishing two} is equivalent to  
\begin{align} E^{<} \subset \ker (w^2 - \Gamma_w). \label{eq:second vanishing} \end{align}
We now make two case distinctions.
\begin{enumerate}
    \item ($E^{<} = \{0\}$) First, suppose that $E^{<} = \{0\}$. This happens exactly when either $N\leq 2$ and $\beta = \pi$ or $N \leq 1$ and $\beta \neq \pi$. In the latter case, the space $E^{\leq} = E$ is one-dimensional, so by~\eqref{eq:first vanishing}, $\Gamma_w = A P_E$ for some constant $A>0$. Thus by the trace condition $\Gamma_w(s,s) = Nw^2(s)$, we must have constant $w$. On the other hand, in the case $N\leq 2$ and $\beta = \pi$, we have 
\begin{align} E^{\leq } = E = \mathrm{span} \{ \ee_{1/2}, \ee_{-1/2}\}.  \label{eq:E lowest case} \end{align}
Hence, the condition $\Gamma_w(s,s) = N w^2(s)$ forces $w^2$ to have Fourier modes only on the frequencies $2\pi$, $-2\pi$ and $0$. Since $w$ is non-negative, this implies that $w$ is of the form~\eqref{eq:weight_formula}. 

\item ($E^{<} \neq \{0\}$) In this case,~\eqref{eq:first vanishing} and~\eqref{eq:second vanishing} implies that $w^2(E^{<}) \subset E^{\leq}$. Hence, 
\begin{align}
\inner{\ee_j, w^2 \ee_k} = \inner{\ee^{\ii 2\pi (j-k) s}, w^2} = 0, \quad \text{for $|j|^2< \frac{\lambda_F}{4\pi^2}< |k|^2$, $j,k\in \frac{\beta}{2\pi} + \Z$.}\label{eq:Fourier vanishing}
\end{align}
In the case $\beta \neq 0, \pi$, the differences $k-j$ run through all integers except $\{0,1\}$ or $\{-1,0\}$ (depending on $\beta$ and $\lambda_F$). Hence, all but these two Fourier modes of $w^2$ vanish, and, since $w^2$ is real-valued, it must be constant. On the other hand, in the case $\beta = 0$ or $\beta = \pi$, the differences $k-j$ in~\eqref{eq:Fourier vanishing} run over all integers except $\{-1,0,1\}$. Since $w^2$ is non-negative, we must have $w$ of the form~\eqref{eq:weight_formula}. 
\end{enumerate}
In summary, we have shown that for all $N$ and $\beta$ equality can only hold if
\[ w^2(s) = A + B\cos(2\pi (s-s_0)) ,\]
and, moreover, that $B= 0$ in the cases $\beta \not\in \{0,\pi\}$ or $\beta = 0$ and $N\leq 1$. Hence, Step 1 completes the proof in these cases.

\textbf{Step 3: identifying the non-constant weight cases.}  Therefore, we can now assume $w^2$ to be of the form above and investigate the density matrices $Q$ that saturate the inequality in the cases $\beta = \pi$ or $\beta =0$ and $N>1$. Note that the special feature about these cases is that $E$ is two-dimensional and given by symmetric Fourier modes:
\begin{align} E = \mathrm{span} \{ \ee_{\pm k} \}, \quad \text{where $k= \frac{\sqrt{\lambda_F}}{2\pi}$.} \label{eq:E representation} \end{align}
Next, we note that $w^2$ acts on the twisted Fourier basis as
\begin{align} w^2 \ee^\beta_j = A\ee^\beta_j + \frac{B}{2}\ee^{-\ii 2\pi s_0} \ee^\beta_{j+1} + \frac{B}{2} \ee^{\ii 2\pi s_0} \ee^\beta_{j-1}.\label{eq:multiplication by w^2}  \end{align}
Hence, a direct calculation yields
\begin{align*} \left(P^{\leq}_\beta w^2 P^{\leq}_\beta\right)(s,s) &= A \dim E^{\leq} + B (\dim E^{\leq}-1) \cos(2\pi (s-s_0)) \\
&= (\dim E^{\leq}-1) w^2(s) + A. \end{align*}
Thus, defining
\[ H\coloneqq \Gamma_w - P^\leq_\beta w^2 P^{\leq}_\beta \]
and using that $\Gamma_w(s,s) = Nw(s)^2$, we find
\begin{align}
H(s,s) =  (N+1 - \dim E^{\leq}) w^2(s) - A. \label{eq:H diagonal}
\end{align}
On the other hand, by~\eqref{eq:first vanishing} and~\eqref{eq:second vanishing}, the self-adjoint operator $H$ annihilates $E^{<}$ and therefore acts only on the two-dimensional eigenspace $E$. However, from the representation in the basis of $E$ in~\eqref{eq:E representation}, the kernel $H(s,s)$ of any operator acting on $E$ has non-vanishing Fourier modes at most at the frequencies $\{-2\sqrt{\lambda_F}, 0, 2\sqrt{\lambda_F}\}$. Consequently, if $N+1 \neq \dim E^{\leq}$ and $\lambda_F \neq \pi^2$, equation~\eqref{eq:H diagonal} can only hold if $w$ is constant. From the explicit spectrum of $H_\beta$, this reduces the problem to either the exceptional integer cases $\beta/\pi + N$ even (where $\dim E^\leq = N+1$) or the lowest anti-periodic cases $\beta = \pi$ and $N\leq 2$ (where $\lambda_F = \pi^2$). The endpoint $N=2,\beta=\pi$ will be shown below to force a constant weight.

\textbf{Step 4: determining $Q$ for exceptional integer cases.} We now focus on the case $N + \beta/\pi$ even and $N\geq 2$. In this case, $N$ is an integer and $\lambda_F = \pi^2 N^2$. Moreover, we have $\dim E^{\leq} = N+1$ and $H(s,s) = -A$. Thus $H$ is of the form
\[ H = -\lambda \ketbra{\ee_{N/2}}{\ee_{N/2}} + (\lambda-A) \ketbra{\ee_{-N/2}}{\ee_{-N/2}}. \]
In particular, $\Gamma_w = P^\leq w^2 P^\leq + H$ is of the form
\begin{align}
\Gamma_w(\lambda)
\coloneqq&
\lambda |\ee_{-N/2} \rangle \langle \ee_{-N/2} |
+(A-\lambda)\ketbra{\ee_{N/2}}{\ee_{N/2}}
+A P^<_\beta
\nonumber\\
&\quad
+
\sum_{j=-N/2+1}^{N/2}
\left(
\frac{B}{2}\ee^{-\ii 2\pi s_0}\,\ketbra{\ee_j}{\ee_{j-1}}
+
\frac{B}{2}\ee^{\ii 2\pi s_0} \,\ketbra{\ee_{j-1}}{\ee_j}
\right)
\label{eq:Gamma-a}
\end{align}
for some $\lambda \in \R$, where the sums run over half-integers for $N$ odd and integers for $N$ even. The remaining issue is to determine the values of $\lambda \in \R$ such that $0\leq \Gamma_w(\lambda) \leq w^2$ and show that this can be written as in~\eqref{eq:Q formula}. For this, first note that, for $f_+$ defined as in~\eqref{eq:f formula},
\[
f_+\ee_k
=
p \ee^{\ii \pi s_0} \ee_{k-1/2}+q \ee^{-\ii \pi s_0} \ee_{k+1/2}.
\]
Therefore, splitting the sums and shifting the indices,
\begin{align*}
\Gamma_+
:=&
f_+ P^\leq_{\pi-\beta} \overline f_+ 
=
\sum_{|k|\leq \frac{N-1}{2}}
\ketbra{p \ee^{\ii \pi s_0} \ee_{k-1/2}+q \ee^{-\ii \pi s_0}  \ee_{k+1/2}}
{p \ee^{\ii \pi s_0}  \ee_{k-1/2}+q \ee^{-\ii \pi s_0} \ee_{k+1/2}} \\
=& p^2 \ketbra{\ee_{-N/2}}{\ee_{-N/2}} + q^2 \ketbra{\ee_{N/2}}{\ee_{N/2}} + \sum_{|k| \leq \frac{N-2}{2}}  (p^2  + q^2)\ketbra{\ee_{k}}{\ee_k} \\
& + pq \sum_{|k|\leq \frac{N-1}{2}} \left(\ee^{\ii 2\pi s_0} \ketbra{\ee_{k-1/2}}{\ee_{k+1/2}} + \ee^{-\ii 2\pi s_0} \ketbra{\ee_{k+1/2}}{\ee_{k-1/2}}\right). 
\end{align*}
Using~\eqref{eq:alpha beta def}, the above expression becomes exactly \eqref{eq:Gamma-a} with $\lambda=p^2 $. A similar calculation shows that
\[
\Gamma_-:= f_- P^\leq_{\pi-\beta} \overline{f_-} = \Gamma_w(q^2). 
\]
In particular, $\Gamma_w(p^2) = \Gamma_+ \geq 0$ and $\Gamma_w(q^2) = \Gamma_- \geq 0$. On the other hand, since
\[ |f_\pm(s)|^2 = p^2 + q^2 + 2pq \mathrm{Re} \, \ee^{\ii 2\pi (s-s_0)} = w(s)^2, \quad \text{and}\quad P^\leq_{\pi-\beta} \leq 1, \]
we also have the upper bound $\Gamma_\pm \leq w^2$. Therefore, $0\leq Q_\pm \leq 1$, and the same holds for any convex combination of these density matrices. This shows that any convex combination of $Q_+$ and $Q_-$ is admissible and saturates the inequality with the weight function in~\eqref{eq:weight_formula}.

To complete the proof in this case, it remains to show that convex combinations of $Q_\pm$  give all admissible solutions of~\eqref{eq:weighted_saturated}. We shall do this by showing that $\Gamma_w(\lambda)$ is no longer non-negative for $\lambda$ outside the interval $[q^2,p^2]$. So first, we note that, in the ordered Fourier basis $\{\ee_{N/2}, \ee_{N/2-1},..., \ee_{-N/2}\}$ of $E^\leq$, $\Gamma_w(\lambda)$ can be written in the matrix form
\begin{align}
    \Gamma_w(\lambda) = \begin{pmatrix} A-\lambda & c & 0 & 0 & \cdots &   0\\
    \overline c & A & c & 0 &\cdots &   0\\
     0 & \overline c & A & c & 0 &\vdots  \\
     \vdots & \ddots & \ddots & \ddots & \ddots & \vdots\\ 
      & \dots & &  \overline c & A & c\\ 
     0 & \dots &  &  0  & \overline c & \lambda \end{pmatrix}, \qquad c = \frac{B}{2} \ee^{-\ii 2\pi s_0}.
\end{align}
Thus 
\[ \det \Gamma_w(\lambda) = -\lambda ^2 \det M + C \lambda + D\]
for some constants $C,D\in \R$ where the minor $M\in \C^{(N-1) \times (N-1)}$ is given by
\begin{align*}
 M = \begin{pmatrix}   A & c & 0 &\cdots &   0\\
      \overline c & A & c & 0 &\vdots  \\
     0 & \ddots & \ddots & \ddots & \vdots \\ 
     \vdots &  & \ddots &\ddots & c \\
     0 & \dots & &  \overline c & A  \end{pmatrix}.
\end{align*}
Hence, by doing Gauss elimination in $M$ to remove the upper triangular part, the diagonal elements of the resulting matrix can be recursively defined as
\begin{align*}
    p_k = A - \frac{|c|^2}{p_{k-1}}, \quad k\geq 2, \quad p_1 = A.
\end{align*}
Hence, using that $|c|^2 = \frac{B^2}{4} \leq \frac{A^2}{4}$ and induction, we find that
\begin{align*}
    p_k \geq \frac{A}{2}>0, \quad \mbox{for any $1\leq k \leq N-1$.}
\end{align*}
In particular, $\prod_k p_k > 0$, and therefore, $\det M > 0$. Consequently, the function $\lambda \mapsto \det \Gamma_w(\lambda)$ is a non-trivial polynomial of degree $2$. On the other hand, we have $\dim E^{\leq} = N+1$ whereas $\rank \Gamma_\pm = N$; therefore,
\[ \det \Gamma_w(p^2) =\det \Gamma_+ =  0 = \det \Gamma_w(q^2). \]
Thus, if $p^2 > q^2$, then $\det \Gamma_w(\lambda) = -\det M (\lambda - p^2)(\lambda-q^2)$, and therefore, $\Gamma_w(\lambda)\geq 0$ if and only if $\lambda \in [q^2,p^2]$. The case $p^2 = q^2$ then follows by a continuity argument.

\textbf{Step 5: determining $Q$ in the lowest anti-periodic case.} Finally, we address the last case $\beta = \pi$ and $N\leq 2$. Here we know that $E^\leq = E$ is given by~\eqref{eq:E lowest case}. So $\Gamma_w$ can be written as
\begin{align*}
\Gamma_w = a \ketbra{\ee_{1/2}}{\ee_{1/2}} + c\ketbra{\ee_{1/2}}{\ee_{-1/2}} + \overline c\ketbra{\ee_{-1/2}}{\ee_{1/2}} + b \ketbra{\ee_{-1/2}}{\ee_{-1/2}}. \end{align*}
The constraint $\Gamma_w(s,s) = N w^2(s) =N (A  + B \cos(2\pi (s-s_0)))$ now forces $a+b = NA$ and $c = NB\ee^{-2\ii \pi s_0}/2$. On the other hand, one can check that 
\begin{align*}
\Gamma_+ = N \ketbra{f_+}{f_+} = w N Q_+ w \quad \text{and} \quad \Gamma_- = N\ketbra{f_-}{f_-} = N w Q_- w
\end{align*} 
satisfy this constraint. Thus any $\Gamma_w$ satisfying this constraint is of the form $\Gamma_w(t) = t \Gamma_+ + (1-t) \Gamma_-$ for $t\in \R$ (in the case $p\neq q$). Moreover, using~\eqref{eq:alpha beta def}, straightforward calculation yields
\[ \det \Gamma_w(t) = N^2 t(1-t) \left(A^2 - B^2\right). \]
Hence, $\Gamma_w(t) \geq 0$ if and only if $t\in [0,1]$. In case $p=q$, $\Gamma_+ = \Gamma_-$ and one can show that this is the only non-negative operator satisfying the diagonal constraint.

It remains to determine when 
\[ Q(t) = tN \ketbra{u_+}{u_+} + (1-t) N \ketbra{u_-}{u_-} \quad u_{\pm} = f_{\pm}/w , \]
satisfies $Q(t) \leq 1$. For $N\leq 1$, this is clearly satisfied for any $t\in [0,1]$. For $N>1$, let $\nu_1,\nu_2$ denote the two non-negative eigenvalues of $Q$. Then
\[
    \nu_1+\nu_2=N,
    \qquad
    \nu_1\nu_2
       =N^2t(1-t)
          \bigl(1-|\langle u_+,u_-\rangle|^2\bigr).
\]
Since $\nu_1 + \nu_2 = N\leq 2$, both eigenvalues are at most one if and only if
\[
    (1-\nu_1)(1-\nu_2)\geq0,
\]
or equivalently, 
\[
    t\in[0,1],
    \qquad
    N^2t(1-t)
          \bigl(1-|\langle u_+,u_-\rangle|^2\bigr) \geq(N-1)_+,
\]
which is precisely~\eqref{eq:t restriction} if we can show that
\[
    |\langle u_+,u_-\rangle|^2=\frac{A-\sqrt{A^2-B^2}}{A+\sqrt{A^2-B^2}}.
\]
This follows from the identities, 
\[
    f_+^{\,2}(s)
       =B + A\cos (2\pi (s-s_0)) -\ii (p^2-q^2) \sin (2\pi (s-s_0)).
\]
and using the integral
\[
    \int_0^1\frac{ds}{A+B\cos(2\pi(s-s_0))}=\frac{1}{\sqrt{A^2-B^2}} , \quad \text{in the case $A>B$.}
\]
The case $A=B$ is trivial because $f_{+}(s)^2 = w(s)^2$.

Finally, if $N=2$,~\eqref{eq:t restriction} forces $\sqrt{A^2-B^2}=A$ and $t=1/2$, thus $B=0$ and
\[
    Q=|\ee_{-1/2}\rangle\langle \ee_{-1/2}|
      +|\ee_{1/2}\rangle\langle \ee_{1/2}|,
\]
which belongs to the constant-weight case.
\end{proof}

\subsection{Inverse cumulative transformation for optimal weights}

We now compute the density profiles with cumulative densities given by the optimal family of weights in~\eqref{eq:weight_formula}. 
\begin{lemma}[Density profiles from optimal weights]
\label{lem:inverse cumulative optimizers}
Let $\Omega=\R$ or $\Omega=\T$, and let $\sqrt{\rho}\in\mH^1(\Omega)$ be a normalized density profile. Suppose that its cumulative density $a$, defined in~\eqref{eq:a def}, satisfies
\begin{align}
    a(s)^2=A+B\cos(2\pi(s-s_0)),
    \qquad A>0,\quad 0\leq B\leq A,\quad s_0\in[0,1).
    \label{eq:general cumulative optimizer}
\end{align}
Then, up to translations, the following hold:
\begin{enumerate}[label=(\roman*)]
\item If $\Omega=\R$, then necessarily $B=A$ and $s_0=1/2$, and
\begin{align}
    \rho(x)=\rhor{\mu}(x)
    =\frac{\mu}{\pi\cosh(\mu x)},
    \qquad \mu=\pi\sqrt{2A}>0.
\end{align}
Moreover, the corresponding cumulative coordinate is
\begin{align}
    S(x)=\frac2\pi\arctan(\ee^{\mu x}). \label{eq:inverse cumulative line}
\end{align}

\item If $\Omega=\T$, then necessarily $B<A$, and
\begin{align}
    \rho(x)=\rhot{m}(x)
    =\frac{2K(m)}{\pi}
    \sqrt{1-m\sin^2\left(
    \mathrm{am}(2K(m)x\,|\,m)\right)},
    \qquad m=\frac{2B}{A+B}\in[0,1).
\end{align}
Moreover, choosing $x_0\in \T$ in~\eqref{eq:S def} such that $\rho(x_0) = \norm{\rho}_\infty$, we have
\begin{align}
    a(s) = \frac{2K(m)}{\pi} \sqrt{1-m \sin^2(\pi s)} \quad \text{and}\quad S(x)=\frac1\pi\mathrm{am}(2K(m)x\,|\,m). \label{eq:inverse cumulative torus}
\end{align}
\end{enumerate}
Conversely, every density profile of the above forms has a cumulative density satisfying~\eqref{eq:general cumulative optimizer}.
\end{lemma}
\begin{proof}
We first consider $\Omega=\R$. By Proposition~\ref{prop:structure cumulative density}, the cumulative density satisfies $a\in\mH^1_0(0,1)$. Thus $a(0)=a(1)=0$, which, together with~\eqref{eq:general cumulative optimizer}, implies $B=A$ and $s_0=1/2$. Consequently,
\[
a(s)=\sqrt{2A}\sin(\pi s)
=\frac{\mu}{\pi}\sin(\pi s),
\qquad \mu=\pi\sqrt{2A}.
\]
In particular, $a>0$ on $(0,1)$ and, by Proposition~\ref{prop:structure cumulative density}, $S:\R \rightarrow (0,1)$ is a $C^1$-diffeomorphism with inverse $X = S^{-1}: (0,1) \rightarrow \R$ satisfying $X'(s) = 1/a(s)$. Moreover, after possibly translating $\rho$, we can assume that $X(1/2) = 0$. Hence,
\[
X(s)=\frac{\pi}{\mu}\int_{1/2}^s
\frac{\dd t}{\sin(\pi t)}
=\frac1\mu\log\tan\left(\frac{\pi s}{2}\right).
\]
Inverting this expression gives
\[
S(x)=\frac2\pi\arctan(\ee^{\mu x}),
\]
and hence
\[
\rho(x)=a(S(x))
=\frac{\mu}{\pi\cosh(\mu x)}.
\]

Next, let $\Omega=\T$. Since the origin of the cumulative coordinate can be chosen arbitrarily in $\T$, we may assume, up to translation, that $s_0=0$. We now claim that $B<A$. Indeed, if $B=A$, then $a$ vanishes only at $s=1/2$ and, by Proposition~\ref{prop:structure cumulative density}~\ref{it:structure cumulative coordinates}, we can find an open interval $J\subset \T\setminus\{x_0\} \cong (0,1)$ such that $X:(0,1/2) \rightarrow J$ is a $C^1$-diffeomorphism with $X'(s) = 1/a(s)$. In particular
\[ |X(s) - X(t)| = \left|\int_s^t \frac{1}{a(u)} \dd u \right| \leq 1, \quad s,t \in (0,1/2). \]
However, $1/a(u)$ is not integrable around $u=1/2$, which yields a contradiction.

Since $A>B$, we can write
\[
a(s)=\sqrt{A+B}\sqrt{1-m\sin^2(\pi s)},
\qquad m=\frac{2B}{A+B}\in[0,1).
\]
Moreover, note that $a>0$ everywhere so that $X$ is a diffeomorphism in $(0,1)$. The normalization $X(1)-X(0)=1$ then gives
\[
1=\int_0^1\frac{\dd s}{a(s)}
=\frac{2K(m)}{\pi\sqrt{A+B}} \quad \text{thus}\quad \sqrt{A+B}=\frac{2K(m)}{\pi}.
\]
and therefore,
\[
a(s)=\frac{2K(m)}{\pi}
\sqrt{1-m\sin^2(\pi s)}.
\]
Choosing $X(0)=0$, we obtain
\[
X(s)=\int_0^s\frac{\dd t}{a(t)}
=\frac1{2K(m)}F(\pi s\,|\,m).
\]
Hence,
\[
S(x)=\frac1\pi\mathrm{am}(2K(m)x\,|\,m),
\]
which yields
\[
\rho(x)=a(S(x))
=\frac{2K(m)}{\pi}
\sqrt{1-m\sin^2\left(
\mathrm{am}(2K(m)x\,|\,m)\right)}.
\]
Finally, the converse statements follow by the same changes of variables.
\end{proof}
\subsection{Proof on the real line}

We now turn to the proof of the gradient-corrected LT inequality in the real line. 

\begin{proof}[Proof of Theorem~\ref{thm:LT_line}]
We can assume $\sqrt{\rho} \in \mH^1(\R)$, as otherwise there is nothing to prove. Then, writing $Q = \sum_j \ketbra{\phi_j}{\phi_j}$ where $Q= U_{\rho/N}^\ast \gamma U_{\rho/N}$, from~\eqref{eq:kinetic identity}, we have
\begin{align} T[\gamma] + 3G(\rho) = \sum_{j} \int_0^1 |(a \phi_j)'(s)|^2 \dd s.\label{eq:first id} \end{align}
Here the normalized density is $\rho/N$, so there is no $N$ in the gradient term as in~\eqref{eq:kinetic identity}.

Moreover, from Proposition~\ref{prop:structure cumulative density}, $a\in \mH^1_0(0,1)$. The density constraint gives
\[
    |\phi_j(s)|^2\leq \sum_k|\phi_k(s)|^2=N
    \quad\text{for a.e. }s\in(0,1),
\]
and hence $|a\phi_j|\leq\sqrt N\,a$ almost everywhere. By~\eqref{eq:first id}, $a\phi_j\in\mH^1(0,1)$. Taking the continuous representatives of $a\phi_j$ and $a$, the preceding inequality extends to every $s\in[0,1]$. Since the trace of $a$ vanishes at both endpoints, we therefore have
\[\lim_{s\uparrow 1} a(s) \phi_j(s) = 0 = \lim_{s\downarrow 0} a(s) \phi_j(s). \]
Consequently $a \phi_j$ belongs to $\mH^1_0(I)$, hence to $\mathcal{Q}_\beta$ for any $-\pi \leq \beta < \pi$. We can therefore apply Lemma~\ref{lem:twisted_filling} in~\eqref{eq:first id} with weight $w(s) = a(s)$ and optimize over $\beta$ to obtain
\begin{align*}
    T[\gamma] \geq \frac{\pi^2}{3} \left(1 + \frac{\kappa(N)}{N^2}\right) \int_\R \rho^3 - 3 G(\rho),
\end{align*}
where
\begin{align*}
    \kappa(N) = \max_{0\leq \beta \leq \pi} \kappa(N,\beta) \quad \mbox{for $\kappa(N,\beta)$ defined in~\eqref{eq:kappa_N_beta}.}
\end{align*}
Note that $\kappa(N,\beta)$ is a positive parabola in $|\beta|$ so that the supremum is achieved at either $\beta = 0$ or $\beta =\pi$. In the integer case, the maximum is achieved at
\begin{align} \beta_N = \begin{cases} \pi, \quad &\mbox{for $N$ odd,} \\ 0, \quad &\mbox{for $N$ even.} \end{cases}\label{eq:beta optimal} \end{align}
This completes the proof of the general version of the inequality in~\eqref{eq:LT_general_line}.

To compute the optimizers for integer $N\in \N$, we note that the only inequality used in the proof was the weighted inequality in Lemma~\ref{lem:twisted_filling}. Hence, equality is achieved if and only if $Q$ saturate the weighted inequality~\eqref{eq:weighted_twisted_filling} with weight $w=a$ and twist $\beta_N$. For $N\geq2$, since $N+\beta_N/\pi$ is even, by Lemma~\ref{lem:inequality saturation}~case~\ref{it:exceptional} this holds if and only if $a=w$ for some $w$ as in~\eqref{eq:weight_formula} and $Q$ is as in~\eqref{eq:Q formula}. 
For $N=1$, use case~\ref{it:lowest}: $a(0)=a(1)=0$ forces $B=A$ and $s_0=1/2$, so $f_+=f_-=a$ and $Q=\ketbra{\ee_0}{\ee_0}$.
By Lemma~\ref{lem:inverse cumulative optimizers}, this is precisely the case when
\begin{align*}
    \rho(x) = N \rhor{\mu}(x) = N \frac{\mu}{\pi \cosh(\mu x)}, \quad \text{for some $\mu >0$,}
\end{align*}
and, by~\eqref{eq:Q formula},~\eqref{eq:f formula}, and~\eqref{eq:inverse cumulative line}, 
\[ \gamma = \sum_{|j|\leq \frac{N-1}{2}} \ketbra{\psi_j}{\psi_j} \quad \text{with}\quad \psi_j(x) = (U_{\rho/N} \ee_j)(x) = \sqrt{\frac{\rho(x)}{N}} \ee^{\ii 4 j \arctan(\ee^{\mu x})}. \]
The equivalent expression in~\eqref{eq:gamma line} simply follows from the exact sum identity
\[ Q(s,t) = \frac{\sin(\pi N (s-t))}{\sin(\pi (s-t))}. \]
\end{proof}

\subsection{Proof on the torus} For the proof of the gradient-corrected LT in the torus, we shall use one additional lemma.
\begin{lemma}[Minimizing phase under flux constraint] \label{lem:optimal phase} Let $a\in \mH^1(0,1)$ with $\min a  >0$, and $\beta \in (-\pi,\pi]$. Then 
\begin{align}
    \inf \left\{ \int_0^1 a^2 |\chi'|^2  : \chi \in \mH^1(I) , \quad  \chi(1) - \chi(0) \in  \beta + 2\pi \Z \right\} = \frac{\beta^2}{\int_0^1 a^{-2}}, 
\end{align}
and the only minimizers (up to an additive constant) are given by $\beta \chi_a(s)$ if $\beta \neq \pi$ and $\pm \pi \chi_a$ if $\beta = \pi$, where
\begin{align}
    \chi_a(s) = \frac{1}{\int_0^1 a^{-2}} \int_0^s a(t)^{-2} \dd t.  \label{eq:optimal phase}
\end{align}
\end{lemma}
\begin{proof}
    By Cauchy-Schwarz,
    \[ |\beta| \leq \left| \int_0^1 \chi'\right| \leq \left( \int_0^1 a^2 |\chi'|^2 \right)^{\frac12} \left(\int_0^1 a^{-2} \right)^{\frac12},\]
    with equality if and only if $\chi' = C a^{-2}$ for some $C\in \R$ and $\chi(1) - \chi(0) = \pm \beta$. Thus $\pm \beta \chi_a$ are the only functions with the desired normalization that saturates this inequality. As $-\beta \in \beta + 2\pi Z$ only if $\beta = \pi$, this completes the proof.
\end{proof}

\begin{proof}[Proof of Theorem~\ref{thm:LT_torus}]
As in the proof of Theorem~\ref{thm:LT_line}, we have~\eqref{eq:first id} with the transformed orbitals $\phi_j$. In this case, however, the weighted orbitals 
\[ g_j(s) \coloneqq a(s) \phi_j(s) \] 
are only (in general) periodic, so that we can only use the estimate for the case $\beta =0$.

To avoid this shortcoming, we use an averaging argument. Precisely, we define
\[ g_{j,\pm \beta} = g_j(s) \ee^{\pm \ii \beta \chi(s)}, \quad \text{for any $\chi \in \mH^1$ with $\chi(1) - \chi(0) = 1$,} \]
and note that, using~\eqref{eq:Q_constraints},
\begin{align}
    \sum_{j} \frac12 \int_0^1 |g_{j,\beta}'|^2 + |g_{j,-\beta}'|^2  = \sum_{j} \int_0^1 |g_j'|^2 + N \beta^2 \int_0^1 a^2 |\chi'|^2. \label{eq:averaging id}
\end{align}
because the phase term cancels out. Therefore, assuming for the moment that $a>0$ everywhere in $(0,1)$ and choosing the phase $\chi_a(s)$ defined in~\eqref{eq:optimal phase}, which optimizes the second term in~\eqref{eq:averaging id} by Lemma~\ref{lem:optimal phase}, and noticing that $g_{j,\pm \beta} \in \mathcal{Q}_{\pm \beta}$ and
\begin{align*}
    \int_0^1 a(s)^{-2} \dd s= N \int_\T \rho^{-1}(x) \dd x, \quad \text{(by change of variables),}
\end{align*}
we can apply Lemma~\ref{lem:twisted_filling} to~\eqref{eq:averaging id} and re-write all integrals back in the original coordinates to obtain
\begin{align}
    T[\gamma] + 3 G(\rho) \geq \frac{\pi^2}{3} \left(1 + \frac{\kappa(N,\beta)}{N^2} \right) \int_\T \rho^3  - \frac{\beta^2}{\int_\T \rho^{-1}}, \quad -\pi \leq \beta \leq \pi. \label{eq:LT torus proof} 
\end{align}

In the case where $a$ vanishes at a point $s_\ast\in (0,1)$, we can instead pick the phase $\chi = 0$ for $s\leq s_\ast$ and $\chi = \beta$ for $s>s_\ast$. Even though this phase has a jump and is not in $\mH^1$, the total functions $g_{j,\pm}(s) = g_j(s) \ee^{\pm \ii \chi}$ are in $\mH^1$ because each $g_j$ vanishes on $s_\ast$. Moreover, the phase term on the right hand-side of~\eqref{eq:averaging id} vanishes, and we can again apply Lemma~\ref{lem:twisted_filling} with $\beta = \pi$ to these terms to obtain~\eqref{eq:LT torus proof} without the shell term. Nevertheless, we note that, in this case, $\int_\T 1/\rho = \int_0^1 1/a^2 = \infty$, because $a$ (or $\sqrt{\rho}$) are $\frac12$-H\"older continuous so that $1/a^2$ diverges at least as $1/|s-s_\ast|$ close to $s_\ast$. In particular, we can keep the full expression in~\eqref{eq:LT torus proof} with the convention that $\cS(\rho) = 1/\infty = 0$ whenever $a$ vanishes at a point.

We now take the maximum over $\beta$ in~\eqref{eq:LT torus proof}, which is achieved again at either $\beta = 0$ and $\beta = \pi$ because the right-hand side is convex in $\beta$. Indeed, the coefficient of $\beta^2$ in~\eqref{eq:LT torus proof} is
\[ \frac{1}{N^2} \int_\T \rho^3 - \cS(\rho) \]
which is non-negative because, by Cauchy-Schwarz,
\[ N = \int_\T \rho \leq \left(\int_\T \frac{1}{\rho} \right)^{\frac12} \left( \int_\T \rho^3 \right)^{\frac12} = \cS(\rho)^{-\frac12} \left(\int_\T \rho^3\right)^{\frac12}. \]
Thus setting $\beta = 0$ or $\beta = \pi$ completes the proof of the general $N>0$ inequality in~\eqref{eq:LT_general_torus}. For integer $N\in \N$, we note that the maximum is achieved at 
\begin{align}
        \beta_N = \begin{dcases} 0, \quad &\mbox{for $N$ even,} \\ \pi, \quad &\mbox{for $N$ odd.} \end{dcases} \label{eq:favorable beta}
\end{align}

To compute the densities and density matrices saturating the inequality, we note that we only used the weighted inequality~\eqref{eq:weighted_twisted_filling} for $g_{j,\pm \beta}$. Hence, the optimizers $\gamma$ of the LT inequality are precisely the density matrices $\gamma$ such that both
\begin{align}  \ee^{\pm\ii \beta_N \chi}U_{\rho/N}^\ast \gamma  U_{\rho/N} \ee^{\mp\ii \beta_N \chi} =  \ee^{\pm\ii \beta_N \chi} Q \ee^{\mp\ii \beta_N \chi}\label{eq:pm density matrices} \end{align}
are optimizer of the weighted inequality~\eqref{eq:weighted_twisted_filling} with $\beta_N$ twist. We can therefore use the characterization in Lemma~\ref{lem:inequality saturation} and the inverse cumulative transformations in Lemma~\ref{lem:inverse cumulative optimizers} to obtain the corresponding formulas in physical space.

Precisely, for $N$ even, we have $\beta_N = 0$ and both density matrices in~\eqref{eq:pm density matrices} are the same. Thus, by Lemma~\ref{lem:inverse cumulative optimizers}~\ref{it:exceptional}, equality holds precisely when $\rho(x) = N \rhot{m}(x)$, $m\in [0,1)$ with $\rhot{m}$ given by~\eqref{eq:density manifold torus} and
\[ \gamma = t \gamma_+ + (1-t) \gamma_- \]
where, by formulas~\eqref{eq:Q formula}, 
\begin{align} \gamma_{\pm} = \sum_{\substack{|k|\leq \frac{N-1}{2} \\ k\in \Z + \frac12}} \ketbra{\psi_{k,\pm}}{\psi_{k,\pm}}, \quad \psi_{k,\pm}(x) = \sqrt{\rhot{m}(x)} \frac{f_{\pm}(S(x))}{a(S(x))} \ee^{\ii 2\pi k S(x)}, \label{eq:pm density matrices optimal} \end{align}
with $a,S$ given by~\eqref{eq:inverse cumulative torus} and $f_{\pm}$ defined as in~\eqref{eq:f formula}. This is precisely~\eqref{eq:LT_torus_equality_gamma} once we note that
\[ \frac{f_{\pm}(s)}{a(s)} = \ee^{\ii \mp \arctan \theta(s)}, \quad  \theta(x) = \frac{p-q}{p+q} \tan( \pi s) = \sqrt{\frac{A-B}{A+B}} \tan(\pi s) = \sqrt{1-m} \tan(\pi s).\]

For $N$ odd, we can again use Lemma~\ref{lem:inequality saturation} and Lemma~\ref{lem:inverse cumulative optimizers} to conclude that equality can only be achieved if $\rho = N \rhot{m}$ for some $m\in [0,1)$, with $\rhot{m}$ given by~\eqref{eq:LT_torus_equality_density}. In this case, however, the density matrices in~\eqref{eq:pm density matrices} are not the same, so we must have both
\begin{align}  \ee^{\pm \ii \pi \chi} Q \ee^{\mp \ii 
\pi \chi} = t_{\pm} Q_+ + (1-t_{\pm}) Q_-, \quad \text{for some $t_\pm \in [0,1]$,} \label{eq:double optimizing identity} \end{align}
with $Q_\pm$ defined as in~\eqref{eq:Q formula}. We now show that this is only possible if 
\[ Q = P^{\leq}_0= \sum_{|j|\leq \frac{N-1}{2}} \ketbra{\ee_j}{\ee_j},\]
which together with the previous formulas yields the optimizer in~\eqref{eq:LT_torus_equality_gamma}.

For this, we first claim that, up to a global phase, 
\begin{align} a(s)\ee^{\ii \pi \chi_a(s)} = f_+(s) \quad \text{or} \quad a(s) \ee^{\ii \pi \chi_a(s)} = f_-(s).\label{eq:claim phase id} \end{align}
Indeed, on the one hand, $|f_\pm(s)| = a(s)$ and $f_\pm$ have minimal kinetic energy among all anti-periodic functions. On the other hand, by Lemma~\ref{lem:optimal phase}, $a(s) \ee^{\ii \pi \chi_a(s)}$ has the minimal kinetic energy among anti-periodic functions $g$ with $|g(s)| = a(s)$. Since there are only two such anti-periodic functions by the uniqueness in Lemma~\ref{lem:optimal phase}, the claim follows. 

Now using identity~\eqref{eq:claim phase id}, we can see that $Q = P^{\leq}_0$ satisfies both~\eqref{eq:pm density matrices optimal} for $t_\pm = 1$ or $t_{\pm} =0$ (depending on which equation in~\eqref{eq:claim phase id} holds). Thus $\gamma$ defined in~\eqref{eq:LT_torus_equality_gamma} is an optimizer. 

To conclude the proof, we could use the explicit expressions for $Q_\pm$ to argue that $Q= P^{\leq}_0$ is the only possible solution of~\eqref{eq:double optimizing identity}. Instead, we shall argue via uniqueness of ground-states, which avoids further calculations. Precisely, we note that any optimizer $\gamma$ minimizes the kinetic energy among all density matrices with density $N\rhot{m}$. Since $\rhot{m}>0$ everywhere in $\T$, by \cite[Theorem 2.4]{Cor26a}, there exists a unique potential $v\in \mH^{-1}(\T)$ such that $\gamma$ is the single-particle density matrix of a fermionic ground-state of $H_N(v) = -\Delta_{\T^N} + \sum_{j=1}^N v(x_j)$. However, since $N$ is odd, by the non-degeneracy theorem in \cite[Theorem 2.3]{Cor26b}, there exists only one such ground-state, and therefore, only one such density matrix. This shows that the optimizer is unique, and therefore given by~\eqref{eq:LT_torus_equality_gamma}. 
\end{proof}

\section{Thomas--Fermi--von Weizs\"acker energy approximation}
\label{sec:TFW energy}

In this section we prove Theorems~\ref{thm:TFW_line} and~\ref{thm:TFW_torus}. To this end, we combine the lower bounds provided by the gradient-corrected LT inequalities with the following upper bounds.

\begin{lemma}[Centered kinetic-energy upper bound]
\label{lem:kinetic_energy_upper_bound}
Let $\sqrt{\rho} \in \mH^1(\R)$ with $\int_\R \rho = N \in \N$. Then we have
\begin{align}
    T_{\rm L}^\R(\rho) \leq \frac{\pi^2}{3} \left(1- \frac{1}{N^2}\right) \int_\R \rho^3 + G(\rho). \label{eq:kinetic_upper_refined_line}
\end{align}
Moreover, let $\rho \in H^1(\T)$ with $\int_\T \rho=N \in \N$.  Then we have
\begin{align}
    T_{\rm L}^{\T}(\rho)
    \leq \frac{\pi^2}{3} \left(1- \frac{1}{N^2}\right) \int_\T \rho^3 + G(\rho) + \mathbf{1}_{N\,\mathrm{even}} \pi^2 \cS(\rho),
    \label{eq:kinetic_upper_refined_torus}
\end{align}
where $\cS(\rho)$ is the shell correction defined in~\eqref{eq:shell correction}. 
\end{lemma}

\begin{proof}
The case $G(\rho)=+\infty$ is trivial; assume $G(\rho)<\infty$.
If $N$ is odd, we can choose the $N$ centered periodic Fourier modes in cumulative coordinates as a trial state. Precisely, we pick as a trial state the density matrix
\begin{align*}
    \gamma := \sum_{|j| \leq \frac{N-1}{2}} \ketbra{\psi_j}{\psi_j}, \quad \text{with}\quad \psi_j(x) \coloneqq \sqrt{\frac{\rho(x)}{N}} \ee^{\ii 2\pi j S(x)}, 
\end{align*}
where we recall that
\begin{align*}
    S(x) = \int_{x_0}^x \frac{\rho(y)}{N} \dd y , \quad \text{and $x_0 \equiv 0$ on $\T$ or $x_0 = -\infty$ on $\R$.}
\end{align*}
Then straightforward calculations show that
\begin{align*}
    T_{\rm L}^\R(\rho) \leq T[\gamma] = \frac{1}{N^3} \sum_{|j| \leq \frac{N-1}{2}} 4\pi^2 j^2 \int_\R \rho^3 + G(\rho) = \frac{\pi^2}{3} \left(1-\frac{1}{N^2}\right) \int_\R \rho^3 + G(\rho),
\end{align*}
which is the desired upper bound. 

For $N$ even we treat the real line and torus separately. In the real line case, we simply define $\phi_j$  via the centered anti-periodic modes, i.e., we set again
\begin{align}
    \gamma := \sum_{|j| \leq \frac{N-1}{2}} \ketbra{\psi_j}{\psi_j}, \quad \text{with}\quad \psi_j(x) = \sqrt{\frac{\rho(x)}{N}} \ee^{\ii 2\pi j S(x)}, \label{eq:test gamma anti-periodic}
\end{align}
but with $j\in \frac12 + \Z$. Then straightforward calculations yields again the desired upper bound. On the torus, if $\rho$ has a zero, choose the cut $x_0$ there. The orbitals have zero endpoint traces and are periodic, giving the bound with $\cS(\rho)=0$. Otherwise, multiply them by a common phase $\ee^{\ii\chi_\rho}$ with $\chi_\rho(1)-\chi_\rho(0)=\pi$ to restore periodicity, choosing
\begin{align*}
    \chi_\rho(x) = \pi \cS(\rho) \int_0^x \frac{1}{\rho(y)} \dd y,
\end{align*}
which is the minimizing phase from Lemma~\ref{lem:optimal phase} in physical coordinates. Thus using $\widetilde \gamma = \ee^{\ii \chi_\rho} \gamma \ee^{-\ii \chi_\rho}$ as a trial state, where $\gamma$ is defined as in~\eqref{eq:test gamma anti-periodic}, we find
\begin{align*}
    T_{\rm L}^{\T}(\rho) \leq T[\widetilde \gamma] &= \int_\T \sum_{|j| \leq \frac{N-1}{2}} \left(2\pi j \frac{\rho(x)}{N} + \pi \frac{\cS(\rho)}{\rho(x)}\right)^2 \frac{\rho(x)}{N} \dd x + G(\rho) \\
    &= \sum_{|j| \leq \frac{N-1}{2}} \frac{4\pi^2 j^2}{N^3}  \int_\T \rho(x)^3  + \pi^2 \cS(\rho)^2 \int_\T \frac{1}{\rho(x)} \dd x  + G(\rho) \\
    &= \frac{\pi^2}{3} \left(1-\frac{1}{N^2}\right) \int_\T \rho^3 + \pi^2 \cS(\rho) + G(\rho),
\end{align*}
where we used in the second identity that the mixed term 
\[ \sum_j 4\pi^2 j \cS(\rho)\int_\T \rho(x)/ N^2 \dd x  \]
vanishes because the sum is centered at zero. This is the desired upper bound. 
\end{proof}

\begin{remark}[Upper bound via lower bound trial states]
We emphasize that, in cumulative-mass coordinates, the trial states used in the proof above are (up to a phase twist) the $Q_\pm$ in Lemma~\ref{lem:inequality saturation}. This explains why the upper and lower bound match exactly whenever $\rho$ is a density coming from optimizers of the LT inequality.
\end{remark}

We now complete the proof of Theorems~\ref{thm:TFW_line} and~\ref{thm:TFW_torus}.

\begin{proof}[Proof of Theorems~\ref{thm:TFW_line} and~\ref{thm:TFW_torus}] In the torus, by using the LT inequality~\eqref{eq:LT_torus} and the upper bound in Lemma~\ref{lem:kinetic_energy_upper_bound}, we have
\begin{align}
   &T_{\rm L}^\T(\rho) \leq \frac{\pi^2}{3} \left(1-\frac{1}{N^2}\right) \int_\T \rho^3 + G(\rho) + \mathbf{1}_{N\,\mathrm{even}} \pi^2 \cS(\rho) ,\\
   &T_{\rm L}^\T(\rho) \geq \frac{\pi^2}{3} \left(1+ \frac{2}{N^2}\right) \int_\T \rho^3 - 3 G(\rho) - \mathbf{1}_{N\,\mathrm{odd}} \pi^2 \cS(\rho).
\end{align}
Hence, by subtracting the TFW plus shell correction functional defined in~\eqref{eq:TFWS def} from the right-hand sides, we obtain the desired energy bounds. The real line case is completely analogous.
\end{proof}

We now prove the following result, which together with Theorems~\ref{thm:TFW_line} and~\ref{thm:TFW_torus} prove Corollary~\ref{cor:exact_manifold}. 

\begin{proposition}[Zero deficit on Poincar\'e density manifold] \label{prop:zero deficit}
Let $\sqrt{\rho} \in \mH^1(\R)$ satisfy $\int_\R \rho =1$, then
\begin{align*}
    \delta_P^\R(\rho) = 4 G(\rho) - \pi^2 \int_\R \rho^3 =0, \quad \text{if and only if} \quad \rho \in \mathcal{D}^\R_P,
\end{align*}
where $\mathcal{D}_P^\R$ is the Poincar\'e density manifold defined in~\eqref{eq:density manifold line}. Similarly, for any $\sqrt{\rho} \in \mH^1(\T)$ with $\int_\T \rho = 1$, we have
\begin{align*}
    \delta_P^\T(\rho) = 4 G(\rho) + \pi^2 \cS(\rho) -  \pi^2 \int_\T \rho^3 =0, \quad \text{if and only if} \quad \rho \in \mathcal{D}_P^\T,
\end{align*}
where $\mathcal{D}_P^\T$ is the periodic Poincar\'e density manifold defined in~\eqref{eq:density manifold torus}.
\end{proposition}

\begin{proof} Let us work out the torus case in detail. The real line case is completely analogous. First, note that, in cumulative coordinates, we have
\begin{align}
    \delta_P^\T(\rho) = 4G(\rho) -\pi^2 \int_\T \rho^3 + \frac{\pi^2}{\int_\T \rho^{-1}} &= \int_0^1 (a')^2 - \pi^2 \int_0^1 a^2 + \frac{\pi^2}{\int_0^1 a^{-2}}. \label{eq:Poincare cumulative coordinate}
\end{align}
If $\min a>0$, the right-hand side can be re-written as
\begin{align*}
    \delta_P^{\T}(\rho) &= \int_0^1 |f_a'|^2 - \pi^2 \int_0^1 |f_a|^2, \quad \text{where}\quad 
    f_a(s) = a \ee^{\ii \chi_a}, \quad \chi_a(s) := \frac{\pi}{\int_0^1 a^{-2}} \int_0^s a(t)^{-2} \dd t.
\end{align*}
If $a(s_\ast)=0$, $s_\ast\in[0,1]$, instead set $f_a=a$ for $s\leq s_\ast$ and $f_a=-a$ otherwise. Then $f_a\in\mathcal Q_\pi$, $\int_0^1 a^{-2}=+\infty$, and the same identity holds.
Hence, by the spectral decomposition of the anti-periodic Laplacian in Lemma~\ref{lem:explicit Laplacian},
\begin{align*}
    \delta_P^\T(\rho) = 0 \quad \text{if and only if} \quad f_a \in \mathrm{span} \{ \ee_{1/2}(s) \coloneqq \ee^{\ii \pi s}, \ee_{-1/2}(s) \coloneqq \ee^{-\ii \pi s}\}.
\end{align*}
If $\min a=0$, equality would force $B=A$ in~\eqref{eq:general cumulative optimizer}, excluded on $\T$ by Lemma~\ref{lem:inverse cumulative optimizers}. Otherwise, by Lemma~\ref{lem:optimal phase}, we know that $f_a$ minimizes the kinetic energy among all anti-periodic functions satisfying $|f(s)|^2 = a(s)^2$. Hence, it suffices to show that
\begin{align*}
    a_\rho^2 \in \{ |f|^2 : f\in \mathrm{span} \{\ee_{1/2}, \ee_{-1/2} \} \} \quad \text{if and only if} \quad \rho \in \mathcal{D}_P^\T.
\end{align*}
However, this set is precisely
\begin{align*}
    \{ |f|^2 : f \in \mathrm{span}\{ \ee_{1/2}, \ee_{-1/2} \} \} = \{ A + B \cos(2\pi (s-s_0)) : s_0 \in [0,1), 0\leq B \leq A\},
\end{align*}
and consequently, the result follows from Lemma~\ref{lem:inverse cumulative optimizers}.
\end{proof}

\section{Local kinetic-energy density}
\label{sec:local kinetic energy} 
In this section, our goal is to prove Theorem~\ref{thm:local kinetic-energy density}. For this, we need to set up some additional notation.

Throughout this section, we let $\Omega = \R$ or $\Omega = \T$ and fix a normalized density profile
\[
    \sqrt{\rho} \in H^1(\Omega),
    \qquad 
    \int_\Omega \rho=1.
\]
We let $\gamma=\sum_j\ketbra{\psi_j}{\psi_j}$ be a minimizer\footnote{The existence of such a minimizer follows by the direct method; see \cite{Lieb83} or \cite[Lemma 4.4]{Cor26a}.} of $T_{\rm L}^{\Omega}(N\rho)$ and denote by $\tau_N(x)$ the associated kinetic-energy density
\begin{align*}
    \tau_N(x) = \sum_{j} |\psi_j'(x)|^2. 
\end{align*}
Then the goal of this section is to compare the exact kinetic-energy density with the Thomas--Fermi-von Weizs\"acker plus shell correction energy density
\begin{align}
    \tau_N^{\rm TFW}(x)
    \coloneqq 
    \frac{\pi^2}{3}N^3\rho^3
    +
    N\left(
        \frac13\rho''
        -
        \frac1{12}\frac{|\rho'|^2}{\rho}
    \right)
    +
    \theta_NN
    \frac{\pi^2\cS(\rho)^2}{\rho}, 
    \label{eq:corrected_physical_kinetic_density}
\end{align}
where here we use the convention that $\cS(\rho) = 0$ if $\Omega = \R$ or if $\rho$ vanishes at a point in $\T$, and we set $\cS(\rho) \rho^{-1} := 0$ on $\{ \rho = 0 \}$.

As in the previous sections, it will be convenient to do this in cumulative coordinates and measure the error in terms of the Poincar\'e deficit. To this end, let us define the modified energy density in cumulative coordinates as
\begin{align}
    e_N(s)
    \coloneqq (-\partial_s a Q a \partial_s)(s,s),\quad \mbox{where} \quad Q = U_\rho^\ast \gamma U_\rho
    \label{eq:energy density}
\end{align}
with $U_\rho$ defined in~\eqref{eq:cumulative isometry def}, and $a$ is the density profile in cumulative coordinates. Moreover, we define its Thomas--Fermi-von Weizs\"acker (plus shell correction) counterpart as
\begin{align}
    e_N^{\rm TFW}
    &\coloneqq  
    \frac{\pi^2}{3}N^3a^2
    +
    N|a'|^2
    +
    \frac N3aa''
    +
    N \theta_N \pi^2 \cS(\rho)^2 a^{-2},
    \label{eq:semiclassical energy density}
\end{align}
where we set $\cS(\rho)a^{-1}:=0$ on $\{a=0\}$ and $a''$ is understood in the distributional sense. Note that the multiplication $a a''$ is well-defined in $\mH^{-1}(0,1)$ because $a\in \mH^1(0,1)$ and $\mH^1(0,1)$ is an algebra.

Let us also define the open set
\begin{align}
    I_a \coloneqq \{ s\in (0,1) : a(s) >0\},
\end{align}
and introduce the space $\mathcal{Y}_a$ of all functions $g:I_a \rightarrow \C$ which are bounded, weakly differentiable, and satisfy
\begin{align*}
    \norm{g}_{\mathcal{Y}_a}^2 \coloneq \norm{g}_\infty^2 + \int_{I_a} |a(s)|^2 |g'(s)|^2 \dd s < \infty.
\end{align*}
In addition, if $a(0) \neq 0$, we require $g$ to be periodic. Then, we define $\mathcal{Y}_a^\ast$ as the dual space endowed with the standard dual norm
\begin{align}
    \norm{F}_{\mathcal{Y}_a^\ast} \coloneqq \sup \left\{ |\inner{F,\phi}| : \phi \in \mathcal{Y}_a, \quad  \norm{\phi}_{\mathcal{Y}_a} \leq 1 \right\}.
\end{align}
As shown in the next lemma, these norms are the cumulative-coordinate versions of the $\mathcal{Y}_\rho$ norms introduced in~\eqref{eq:Yrho norm} and~\eqref{eq:dual Yrho topology}. 

\begin{lemma}[Pull-back of energy densities] \label{lem:Ya topology} Let $\mathcal{Y}_a$ be the Banach space introduced above. Then the pull-back map
\[ g \mapsto (S^\# g)\coloneqq g \circ S \]
is an isomorphism from $\mathcal{Y}_a$ to $\mathcal{Y}_\rho$. Moreover, for any $g\in \mathcal{Y}_a$ we have
\begin{align*}
    &\inner{ \tau_N^{\rm TFW}, (S^\# g)}_x = \left\langle e_N^{\rm TFW}-\frac{3N}{4} |a'|^2,g\right\rangle_s, \\
    &\inner{\tau_N, (S^\# g)}_x = \left\langle e_N-\frac{3N}{4} |a'|^2,g\right\rangle_s,
\end{align*}
where $e_N$ and $e_N^{\rm TFW}$ are defined in~\eqref{eq:energy density} and~\eqref{eq:semiclassical energy density}.
\end{lemma}

\begin{proof} As $S$ is a $C^1$-diffeomorphism from $\Omega_\rho(x_0) \coloneqq \{x \in \Omega\setminus \{x_0\} : \rho(x) > 0 \}$ to $I_a$ by Proposition~\ref{prop:structure cumulative density}~\ref{it:structure cumulative coordinates}, it follows at once from the change of variables formula and the chain rule that $S^\#$ maps $\mathcal{Y}_a$ surjectively to the space of bounded and weakly differentiable functions in $\Omega_\rho(x_0)$ with equalities
\begin{align*} &\sup\{ |g(S(x))| : x\in \Omega_\rho(x_0)\} = \sup \{ |g(s)| : s\in I_a\} \\
\intertext{and}
&\int_{\Omega_\rho(x_0)} \rho(x) |\partial_x (g \circ S)(x)|^2 \dd x= \int_{I_a} a(s)^2 |\partial_s g(s)|^2 \dd s. \end{align*}
Since $x_0 = -\infty$ in the case $\Omega = \R$, this proves the first statement in this case. In the case $\Omega= \T$ with $a(0) = 0$, we have $\Omega_\rho(x_0) = \{x\in \T: \rho(x) > 0 \}$ so the result is also proved. If $a(0) \neq 0$, the periodicity of $g$ implies that $g \circ S$ has the same limit when $x$ approaches $x_0$ from the left and from the right, and therefore $g\circ S$ is also weakly differentiable in $\Omega_\rho(x_0) \cup \{x_0\} = \{ \rho >0 \}$. This completes the proof of the first statement. 

For $g\in\mathcal Y_a$, extend $ag$ by zero on $\{a=0\}$. This belongs to $H^1$, since $|ag|\leq\|g\|_\infty a$ and $(ag)'=a'g+ag'\in L^2$. Define
\begin{align*}
    \langle aa'',g\rangle_s&:=\langle a'',ag\rangle_s=-\int_{I_a}(|a'|^2g+aa'g'),
\end{align*}
which is continuous on $\mathcal Y_a$ by Cauchy--Schwarz. Change of variables then gives
\begin{align*}
    \langle\rho'',g\circ S\rangle_x=-\int_{I_a}aa'g'=\langle aa'',g\rangle_s+\int_{I_a}|a'|^2g,
\end{align*}
yielding the TFW identity.
For the exact energy densities, one can similarly use the identity
\begin{align*}
    e_N &= \sum_j |\partial_s (a^{\frac12} \psi_j \circ X)|^2 = \sum_j \left(\frac{1}{a} |\partial_x \psi_j \circ X|^2 + \frac{a'}{a} \mathrm{Re} \,\overline{\partial_x \psi_j} \psi_j\circ X + \frac14\frac{|a'|^2}{a} |\psi_j|^2 \circ X \right)\\
    &= \frac{1}{a} \tau_N\circ X + \frac{3N}{4} |a'|^2,
\end{align*}
which follows because $\sum_j |\psi_j|^2 = N \rho$ and $a' = \rho'/\rho \circ X$, and apply again the change of variables along the connected components. 
\end{proof}

We now prove an upper bound on the $\mathcal{Y}_a^\ast$-error between the two modified energy densities. This result can be seen as a local Weyl estimate, on which the density rather than the external potential is prescribed.

\begin{lemma}[Density-constrained local Weyl estimate]
\label{lem:weak_local_Fermi}
Let $e_N$ and $e_N^{\rm TFW}$ be defined as in~\eqref{eq:energy density} and~\eqref{eq:semiclassical energy density}.  Then
\begin{align}
    \norm{e_N-e_N^{\rm TFW}}_{\mathcal{Y}_a^\ast}
    \leq
    N\left(\sqrt{\delta_{\rm P}^\Omega(\rho)}
        + \delta_{\rm P}^\Omega(\rho)
        +
        \frac23
        \norm{\mathcal R_{\rm P}^\Omega[\rho]}_{\mathcal{Y}_a^\ast} \right),
    \label{eq:weak_local_Fermi}
\end{align}
where $\mathcal{R}_P^\Omega[\rho]$ is the local Poincar\'e residual
\begin{align}
    \mathcal R_{\rm P}^\Omega[\rho]
    \coloneqq \begin{dcases} 
    \pi^2 \cS(\rho)^2 a^{-2}-\pi^2a^2-aa'',\quad &\mbox{for $\Omega =\T$, } \\
    -\pi^2a^2-aa'', \quad &\mbox{for $\Omega = \R$.} \end{dcases}
    \label{eq:Poincare_residual}
\end{align}
\end{lemma}

\begin{proof}
Let us define the lower local model $\ell_N$ as
\begin{align}
    \ell_N \coloneqq \frac{\pi^2}{3} (N^2+2) N a^2+N(aa')' - N \mathbf{1}_{N \text{ odd}} \pi^2 \cS(\rho)^2 a^{-2},
    \label{eq:local_lower_model}
\end{align}
where, we recall the convention that $\cS(\rho) =0$ if $\Omega = \R$ or if $\rho$ vanishes at a point. Then it suffices to prove that
\begin{align}
    \norm{e_N-\ell_N}_{\mathcal{Y}_a^\ast}
    \leq
    N\left(\sqrt{\delta_{\rm P}(\rho)} +\delta_{\rm P}(\rho) \right).
    \label{eq:lower_model_local_estimate}
\end{align}
Indeed, assuming the above inequality holds, then the proof follows from the triangle inequality and the identity
\begin{align*}
    e_N^{\rm TFW} - \ell_N = \frac{2}{3} N \mathcal{R}_P^\Omega[\rho]. 
\end{align*}

To prove~\eqref{eq:lower_model_local_estimate}, we first show that
\begin{align}
    \inner{e_N- \ell_N,|w|^2} \geq - N \int_0^1 |w'|^2 a^2, \quad \mbox{for any real-valued $w\in \mathcal{Y}_a$.}  \label{eq:weak comparison}
\end{align}
For this, we recall from Lemma~\ref{lem:cumulative_kinetic_identity} that
\[ (-\partial_s a Q a \partial_s)(s,s) = \sum_j |(a \phi_j)'|(s)^2, \]
where $\phi_j$ satisfy $a\phi_j \in \mH^1(0,1)$. We now make two observations. First, for any $w\in \mathcal{Y}_a$, the function $w a \phi_j \in \mH^1_0(0,1)$ if $a\in \mH^1_0(0,1)$ and $w a \phi_j \in \mH^1(\T)$ if $a\in \mH^1(\T)$. Indeed, $w a \phi_j \in \mH^1$ by the product rule and a standard approximation argument, while $w a \phi_j$ vanishes whenever $a = 0$ because $|w a \phi_j| \leq |w| |a| N$ and $w$ is bounded. The periodicity follows once $a$ is periodic because we assumed $w$ to be periodic in this case. Second, from the identities 
\begin{align*}
    \sum_{j} |\phi_j(s)|^2 = N \quad \text{and} \quad 2 \sum_{j} \mathrm{Re} \, \overline{\phi_j'}(s) \phi_j(s) = \partial_s N = 0. 
\end{align*}
we obtain the pointwise (a.e.) identity
\begin{align}
    \sum_j \mathrm{Re} \, \overline{ (a w \phi_j)'} w' (a \phi_j) &= N \mathrm{Re}\, \overline{ (a w)'} w' a - 2 a^2 \sum_j \mathrm{Im}(\overline w w') \mathrm{Im} (\overline \phi_j' \phi_j)\nonumber  \\
    &=N |w' a|^2 + N  \mathrm{Re} \,\overline{w a'} w' a -2a^2 \sum_j \mathrm{Im}(\overline w w') \mathrm{Im} (\overline \phi_j' \phi_j), \label{eq:middle term id}
\end{align}
for every $w\in \mathcal{Y}_a$. Therefore, for real-valued $w\in \mathcal{Y}_a$,
\begin{align}
    \int_0^1 e_N|w|^2   &= \sum_{j} \int_0^1 |(a w \phi_j)' - w'(a\phi_j)|^2\nonumber \\
    &= \int_0^1 \sum_j |(a w \phi_j)'|^2  - N \int_0^1 |w' a|^2 - N \int_0^1 2 \mathrm{Re} \, w'a \overline{wa'}. \label{eq:useful id}
\end{align}
Therefore, in the case $\Omega = \R$, we can apply the weighted inequality in Lemma~\ref{lem:twisted_filling} with periodic boundary conditions for even $N$ and anti-periodic boundary conditions for $N$ odd to the first term in~\eqref{eq:useful id} to obtain~\eqref{eq:weak comparison}. 

In the case $\Omega = \T$ and $N$ even, we can also directly use the periodic weighted inequality in~\eqref{eq:useful id} to obtain~\eqref{eq:weak comparison}. For $N$ odd and $\min a>0$, we do the same averaging trick from previous proofs, i.e., we set
\[ w_{\pm} \coloneqq w \ee^{\pm \ii \chi} \quad \text{where}\quad \chi(s) = \pi \cS(\rho) \int_0^s a^{-2}(t) \dd t, \] 
and then use~\eqref{eq:useful id} (plus the imaginary term from~\eqref{eq:middle term id}) with $w_\pm$ instead of $w$, apply the anti-periodic weighted inequality~\eqref{eq:weighted_twisted_filling}, and average the resulting inequalities so that the current and imaginary terms cancel out. This yields again~\eqref{eq:weak comparison}. Note that, at this step, it is important that we use only periodic test functions $w$ in the case $\Omega = \T$.
If $a(s_\ast)=0$, $s_\ast\in[0,1]$, instead apply the weighted inequality to $h_j(s)=(-1)^{\mathbf1_{\{s>s_\ast\}}}w(s)a(s)\phi_j(s)$. Since $wa\phi_j$ vanishes at $s_\ast$, $h_j\in\mathcal Q_\pi$ and $|h_j'|=|(wa\phi_j)'|$ a.e., giving~\eqref{eq:weak comparison} with $\cS(\rho)=0$.

The idea now is to test~\eqref{eq:weak comparison} against $w_{\pm} = \sqrt{1\pm r}$ to obtain the desired estimate~\eqref{eq:lower_model_local_estimate}. To do this, first note that, since
\begin{align*}
    &\inner{\ell_N, 1} = \frac{\pi^2}{3} (N^2 + 2) N \int_\Omega \rho^3 - N\mathbf{1}_{N \text{ odd}} \pi^2 \cS(\rho),
    \intertext{and}
    &\inner{e_N, 1} = T_{\rm L}^{\Omega}(N \rho) + 3N G(\rho),  
\end{align*}
the upper bound in Lemma~\ref{lem:kinetic_energy_upper_bound} yields
\begin{align}
    E_N \coloneqq \inner{e_N- \ell_N, 1} \leq 4 N G(\rho) - N \pi^2 \int_\Omega \rho^3 + N \pi^2 \cS(\rho) = N \delta_P(\rho). \label{eq:energy estimate}
\end{align}

Let now $r\in \mathcal{Y}_a$ be real-valued with $\norm{r}_{\mathcal{Y}_a}\leq 1$. Then for any $0<t<1$, the function 
\[ w_{\pm} \coloneqq \sqrt{1\pm t r}\] 
belongs to $\mathcal{Y}_a$. Hence, we can apply estimate~\eqref{eq:weak comparison} to $w_\pm$ to obtain
\[
   E_N \pm t \inner{e_N-\ell_N, r} \geq - \frac{N}{4} t^2\int_0^1 \frac{|r' a|^2}{1\pm t r}, 
\]
which together with~\eqref{eq:energy estimate} gives
\begin{align*}
    \left| \inner{e_N-\ell_N,r} \right|
    &\leq \frac{E_N}{t} + \frac{N}{4}
    \frac{t}{1-t} \leq
    N\left[ \frac{\delta_{\rm P}(\rho)}{t} +\frac{t}{4(1-t)}
    \right].
\end{align*}
For $\delta_{\rm P}(\rho)>0$, the right-hand side is minimized at
\[
    t_\ast
    =
    \frac{2\sqrt{\delta_{\rm P}(\rho)}}
    {1
     +2\sqrt{\delta_{\rm P}(\rho)}},
\]
which yields~\eqref{eq:lower_model_local_estimate}. If $\delta_{\rm P}(\rho)=0$, the same conclusion follows by letting
$t\downarrow0$.
\end{proof}

\begin{remark}[Extension to approximate minimizers]
Note that the only place where we used that $e_{N}$ is the exact minimizer was to apply the upper bound in Lemma~\ref{lem:kinetic_energy_upper_bound}. Hence, the same result holds if we assume only this upper bound, which justifies the approximate minimizer version of Theorem~\ref{thm:local kinetic-energy density} in Remark~\ref{rem:approximate minimizers}.
\end{remark}

We now show that the Poincar\'e residual can be controlled by the Poincar\'e deficit.

\begin{lemma}[Poincar\'e residual bound via Poincar\'e deficit] \label{lem:poincare residual bound}
Let $\Omega = \R$ or $\Omega = \T$ and $\sqrt{\rho} \in \mH^1(\Omega)$ be a normalized density profile. Then the Poincar\'e residual $\mathcal{R}_P^\Omega[\rho]$ defined in~\eqref{eq:Poincare_residual} satisfies
\begin{align}
    \norm{\mathcal{R}^\Omega_P[\rho]}_{\mathcal{Y}_a^{\ast}} \leq \left(1 + 2 \sqrt{G(\rho)} + \pi \sqrt{\cS(\rho)}\right) \sqrt{\delta_P^\Omega(\rho)}, \label{eq:Poincare residual bound}
\end{align}
with our convention that $\cS(\rho) = 0$ when $\Omega = \R$ or when $\rho$ vanishes at a point.
\end{lemma}

\begin{proof} Let us present the proof in the torus case. Suppose first that $\min a >0$ and set
\[
    f(s)\coloneqq a(s)e^{i\chi(s)}, \quad \chi(s) = \pi \cS(\rho) \int_0^s a^{-2}(t) \dd t. 
\]
Then $f$ is anti-periodic and
\begin{align}
    \delta_{\rm P}^\T(\rho)
    =
    \inner{f,L_\pi f}, \quad\text{where}\quad L_\pi = -\partial_s^2 - \pi^2. 
    \label{eq:deficit_spectral_form}
\end{align}
Moreover, a direct computation gives
\begin{align}
    \mathcal R_{\rm P}^\T[\rho]
    =
    a\,\Re\!\left[
        e^{-i\chi}
        (-\partial_s^2-\pi^2)f
    \right] = \mathrm{Re} \overline{f} (-\partial_s^2 - \pi^2) f. 
    \label{eq:residual_spectral_form}
\end{align}
Thus, since $L_\pi \geq 0$ by the spectral decomposition of the anti-periodic Laplacian (see Lemma~\ref{lem:explicit Laplacian}), for any real-valued $g\in \mathcal{Y}_a$, it follows from Cauchy-Schwarz that
\begin{align*}
   |\inner{g, \mathcal{R}_P^\T(\rho)}|^2 =  \left|\mathrm{Re} \, \inner{g f, L_\pi f}\right|^2 \leq \inner{g f, L_\pi g f}\,   \delta_P(\rho). 
\end{align*}
To conclude, we can combine the above inequality with the estimate
\begin{align*}
\inner{ f g, L_\pi fg} &= \int_0^1 |(g a)'|^2 + \pi^2 \cS(\rho)^2 (g/a)^2 - \pi^2 (ga)^2 \\
&\leq \left((1+ \norm{a'}_{\mL^2})^2  + \pi^2 \cS(\rho) \right) \norm{g}_{\mathcal{Y}_a}^2,
\end{align*}
take the supremum over real-valued $g \in \mathcal{Y}_a$ with $\norm{g}_{\mathcal Y_a}\leq 1$, and recall that $\norm{a'}_{\mL^2}^2 = 4 G(\rho)$ to obtain the desired bound. The case where $a$ vanishes at some $s_\ast \in (0,1)$ follows by taking $f = a$ for $s\leq s_\ast$ and $f= -a$ otherwise and repeating the same arguments with the convention that $\cS(\rho) = 0$. In case $a$ vanishes at the endpoints of $(0,1)$, we can simply take $f=a$ as $a$ is already anti-periodic.

The case $\Omega = \R$ follows the analogous steps with the simplification that $\mathcal{R}_P^\R[\rho] = a L_\pi a$ so that we can use we $a(s)$ directly instead of $f(s)$.
\end{proof}

We can now complete the proof of Theorem~\ref{thm:local kinetic-energy density}.

\begin{proof}[Proof of Theorem~\ref{thm:local kinetic-energy density}] By the identities in Lemma~\ref{lem:Ya topology} and the estimates in Lemma~\ref{lem:weak_local_Fermi} and Lemma~\ref{lem:poincare residual bound},
\begin{align}
    |\inner{\tau_N - \tau_N^{\rm TFW}, g\circ S}_{x}| = |\inner{e_N- e_N^{\rm TFW}, g}_s| \leq C_\rho \norm{g}_{\mathcal{Y}_a}, \quad g\in \mathcal{Y}_a,\label{eq:main est}
\end{align}
where
\[ C_\rho = N\left(\delta_P(\rho) + \frac
{\sqrt{\delta_P(\rho)}}{3} \left(5+4\sqrt{G(\rho)} + 2\pi \sqrt{\cS(\rho)}\right)\right). \]
Since $S^\#:\mathcal{Y}_a\rightarrow \mathcal{Y}_\rho$ is an isomorphism by Lemma~\ref{lem:Ya topology}, the result follows by taking the supremum over $g$ with $\norm{g}_{\mathcal{Y}_a} \leq 1$. 
\end{proof}

\section{Approximation of the inverse Kohn--Sham potential}
\label{sec:KS potential}

We now turn to the proofs of Theorem~\ref{thm:KS_high_density} and Corollary~\ref{cor:KS_Poincare_manifold}. For this, we shall use the following version of the force-balance equation.  

\begin{lemma}[Stationary force-balance equation]
\label{lem:force_balance}
Let $\rho \in \mathcal{D}_1^\T$, $N\in \N$, $-v_{\rm KS}(N\rho) \in\partial T_{\rm L}^\T(N\rho) \subset \mH^{-1}(\T)/\{1\}$, and $\tau_N$ be the exact kinetic-energy density defined in~\eqref{eq:exact kinetic-energy density}. Suppose that $\rho \in \mH^2(\T)$. Then, 
\begin{align}
     \inner{v_{\rm KS}(N\rho),(N \rho \xi)'}
    = -\inner{ 2 \tau_N - \frac{N}{2} \rho'', \xi'}, \quad \mbox{for every $\xi \in \mH^2(\T;\R)$.}
    \label{eq:force_balance}
\end{align}
Similarly, if we define the inverse TFW potential as
\begin{align}
    v_{\rm TFW}(N\rho)
    \coloneqq
    -\pi^2N^2\rho^2
    -
    \frac13\frac{(\sqrt\rho)''}{\sqrt\rho}
    -
    \theta_N\frac{\pi^2 \cS(\rho)^2}{\rho^2},
    \label{eq:TFW potential}
\end{align}
then we have
\begin{align}
    \inner{v_{\rm TFW}(N\rho),(N\rho \xi)'}
    = -\inner{2 \tau_N^{\rm TFW} - \frac{N}{2} \rho'', \xi'}, \quad \mbox{for every $\xi \in \mH^2(\T;\R)$,}
    \label{eq:TFW force_balance}
\end{align}
where $\tau_N^{\rm TFW}$ is the approximate kinetic-energy density defined in~\eqref{eq:corrected_physical_kinetic_density}.
\end{lemma}

\begin{proof} Fix a periodic real-valued function $\xi \in \mH^2(\T)$.  For $|t|$ sufficiently small, the map
\[
    \Phi_t(x)\coloneqq x+t\xi(x)
\]
is an orientation-preserving diffeomorphism of $\T$ with Jacobian  $J_t(x)\coloneqq\Phi_t'(x)=1+t\xi'(x)$. Then, we can define 
\[ \gamma_t = \sum_j \ketbra{\psi_{j,t}}{\psi_{j,t}} \quad \text{via the pull-back}\quad \psi_{j,t}(x) \coloneqq \frac{\psi_j(\Phi_t^{-1}(x))}{\sqrt{J_t(\Phi_t^{-1}(x))}}. \]
Note that $\psi_{j,t} \in \mH^1(\T)$ and a direct change of variables gives $\inner{\psi_{j,t},\psi_{k,t}} = \delta_{jk} \norm{\psi_j}^2$ so $\gamma_t$ is an admissible trial state with $\Tr \gamma_t = \Tr \gamma = N$ and density
\begin{align*}
    N \rho_t(x) = N \frac{\rho(\Phi_t^{-1}(x))}{J_t(\Phi_t^{-1}(x))}. 
\end{align*}
Moreover, differentiating at $t=0$ gives
\begin{align}
    \lim_{t\rightarrow 0} \frac{\rho_t- \rho_0}{t} =  - (\rho \xi)',
    \label{eq:density_first_variation}
\end{align}
where the limit is taken in the $\mH^1(\T)$ sense. This follows from a standard Taylor expansion plus dominated convergence argument. Moreover, we emphasize that, at this step, it is crucial that $\rho \in \mH^2(\T)$. 

We now compute the kinetic-energy variation.  From the definition of $\psi_{j,t}$,
\[
    \psi_{j,t}'(\Phi_t(x))
    =
    \frac{\psi_j'(x)}{J_t(x)^{3/2}}
    -
    \frac12
    \frac{\psi_j(x)J_t'(x)}{J_t(x)^{5/2}}.
\]
Hence, after the change of variables $y=\Phi_t(x)$,
\begin{align*}
    \int_\T|\psi_{j,t}'(y)|^2\,\dd y
    &=
    \int_\T\frac{|\psi_j'|^2}{J_t^2}
    -
    \int_\T\frac{J_t'}{J_t^3}
    \Re\!\left(\overline{\psi_j}\psi_j'\right)
    +
    \frac14
    \int_\T\frac{(J_t')^2}{J_t^4}|\psi_j|^2 .
\end{align*}
Since $J_t=1+t\xi'$ and $J_t'=t\xi''$, we obtain
\begin{align}
    \left.
    \frac{\dd}{\dd t}
    \right|_{t=0}
    \Tr(-\partial_x^2\gamma_t)
    &=\sum_{j} \left.
    \frac{\dd}{\dd t}
    \right|_{t=0}
    \int_\T|\psi_{j,t}'|^2
    =
    -\sum_j \int_\T2\xi'|\psi_j'|^2
    + \xi''
    \Re\!\left(\overline{\psi_j}\psi_j'\right)
    \nonumber \\
    &=
    -2\int_\T\xi'\tau_N
    -
    \frac12\int_\T\xi''(N \rho)' = -\inner{2 \tau_N - \frac{N}{2} \rho'', \xi'}. \label{eq:kinetic first variation}
\end{align}
Now, since $-v_{\rm KS}(N\rho) \in \partial T_{\rm L}^\T(N\rho)$, the subgradient inequality gives
\[
    T_{\rm L}^\T(N \rho_t)
    -
    T_{\rm L}^\T(N \rho)
    \geq -
    N\inner{v_{\rm KS}(N\rho) ,\rho_t-\rho}.
\]
Thus, since $\gamma_t$ is an admissible trial state for $T_{\rm L}^\T(N \rho_t)$, while $\gamma$ minimizes $T_{\rm L}^\T(N\rho)$, 
\[
    \Tr(-\partial_x^2\gamma_t)
    -
    \Tr(-\partial_x^2\gamma) \geq - N \inner{v_{\rm KS}(N\rho), \rho_t - \rho}
\]
Now dividing by $t$, taking the limit $t\rightarrow 0$ from below and above and using~\eqref{eq:density_first_variation} and~\eqref{eq:kinetic first variation} gives \eqref{eq:force_balance}.

To prove equation~\eqref{eq:TFW force_balance}, note that straightforward calculation yields
\[ (\rho^2)' \rho = \frac23 (\rho^3)', \quad \frac13 \left(\frac{(\sqrt{\rho})''}{\sqrt{\rho}}\right)' \rho = \left(\frac{\rho''}{6} - \frac16 \frac{(\rho')^2}{\rho}\right)', \quad \text{and}\quad \left(\frac{1}{\rho^2}\right)' \rho = 2\left(\frac{1}{\rho}\right)', \]
with multiplications well-defined in $\mH^{-1}(\T)$ under the assumption $\rho \in \mH^2(\T)$. Hence,
\begin{align*} v_{\rm TFW}(N\rho)' N\rho &= - \frac{2 \pi^2}{3} N^3 (\rho^3)' - N \left(\frac{\rho''}{6} - \frac16 \frac{(\rho')^2}{\rho}\right)' - \theta_N 2 \pi^2 N\cS(\rho)^2 \left(\frac{1}{\rho}\right)'\\
&= -\left(2\tau_N^{\rm TFW} - \frac{N}{2} \rho''\right)',\end{align*}
which completes the proof. 
\end{proof}

We can now complete the proof of Theorem~\ref{thm:KS_high_density}.

\begin{proof}[Proof of Theorem~\ref{thm:KS_high_density}]
Let $\eta \in \mH^1(\T)$ with $\int_\T \eta =0$ and define
\begin{align*}
    \xi(x) \coloneqq \frac{1}{\rho(x)} f(x) ,\qquad f(x) \coloneqq \int_0^x \eta(y) \dd y. 
\end{align*}
Then, 
\[ \xi' = \frac{\eta}{\rho} - \frac{\rho'}{\rho^2} f \quad \text{and}\quad \xi'' = \frac{\eta'}{\rho} - \frac{2\rho'}{\rho^2} \eta +\left(2 \frac{(\rho')^2}{\rho^3} - \frac{\rho''}{\rho^2} \right) f. \]
We can now easily estimate the $\mathcal{Y}_\rho$ norm of $\xi'$ in terms of $\rho$ and the $\mH^1$-norm of $\eta$. To obtain an estimate with as few powers of $m \coloneqq \norm{1/\rho}_\infty$ as possible,  we first note that
\begin{align*}
    \int_\T \rho \left(2 \frac{(\rho')^2}{\rho^3} - \frac{\rho''}{\rho^2} \right)^2 \left(\int_0^x \eta\right)^2 \dd x &\leq \norm{\eta}_2^2 \int_\T \rho \left(2 \frac{(\rho')^2}{\rho^3} - \frac{\rho''}{\rho^2} \right)^2 \\
    &= \norm{\eta}_2^2 \left(\int_\T \frac{(\rho'')^2}{\rho^3} - \frac43 \int_\T \frac{(\rho')^4}{\rho^5}\right)  \\
    &\leq  m^3 \norm{\eta}_2^2 \norm{\rho''}_2^2, 
\end{align*}
where we used integration by parts with the equality $3(\rho')^2 \rho'' = ((\rho')^3)'$ to obtain the second identity. From this inequality, the triangle inequality and H\'older's inequality, 
\begin{align}
    \norm{\xi'}_{\mathcal{Y}_\rho} = \left(\norm{\xi'}_{\mL^\infty}^2 + \int_\T \rho |\xi''|^2\right)^{\frac12} 
    \leq& m \norm{\eta}_\infty + m^2 \norm{\rho'}_\infty \norm{\eta}_2 + m^{\frac12} \norm{\eta'}_2 \nonumber \\
    &+ m^{\frac32} \norm{\eta}_2 \left(2 \norm{\rho'}_\infty + \norm{\rho''}_2\right) \nonumber \\
    \leq& m^2 \left(4 C_{\rm Sob} +2\right) \norm{\rho}_{\mH^2} \norm{\eta}_{\mH^1},  \label{eq:some est}
\end{align}
where we used the estimates $m\geq 1$ and $\norm{\rho}_{\mH^2} \geq \norm{\rho}_{\mL^2} \geq 1$ (which holds because $\int_\T \rho = 1$), to simplify the estimate in the last line. 

To conclude, note that $(\rho \xi)' = \eta$, and therefore, by the force-balance equation in Lemma~\ref{lem:force_balance}, the local kinetic energy estimate in Theorem~\ref{thm:local kinetic-energy density} and~\eqref{eq:some est},
\begin{align*}
    |\inner{v_{\rm KS}(N\rho)-v_{\rm TFW}(N\rho), \eta }| &= \frac{2}{N} |\inner{\tau_N - \tau_N^{\rm TFW}, \xi'}|  \\
    &\leq D_\rho \left(\sqrt{\delta^\T_P(\rho)} + \delta^\T_P(\rho)\right) \norm{\eta}_{\mH^1} ,
\end{align*}
where $D_\rho$ is the constant in Theorem~\ref{thm:KS_high_density}. Hence, taking the supremum over all $\eta$ with $\norm{\eta}_{\mH^1}\leq 1$ and $\int_\T\eta =0$ completes the proof. 
\end{proof}

We now prove Corollary~\ref{cor:KS_Poincare_manifold}.

\begin{proof}[Proof of Corollary~\ref{cor:KS_Poincare_manifold}]
For the periodic case, recall that
\[
    \mathrm{sn}(x | m ) = \sin(\mathrm{am}(x  | m ) ), \quad \mathrm{cn}(x|m) = \cos(\mathrm{am}(x| m)),\]
and
\[ \mathrm{dn}(x \, | \, m):=\sqrt{1-m \sin^2 \left(\mathrm{am}\big(x\,|\,m\big)\right)}.
\]
Then, using the identity 
\[ \mathrm{am}'(x\, | m) = \mathrm{dn}(x \, | \, m)  \quad \text{(which follows from the definition of $\mathrm{am}$)},\]
we find
\begin{align*} 
\mathrm{dn}'(x \, | \, m) = -m\, \mathrm{sn}(x|m)\, \mathrm{cn}(x|m) \quad \text{and}\quad \mathrm{dn}''(x\, | \, m) = m \, \mathrm{dn}(x\, | \, m) \left(\mathrm{sn}^2(x|m) - \mathrm{cn}^2(x|m)\right). \end{align*} 
Thus, since $\rhot{m}(x) = 2K(m)/\pi \,  \mathrm{dn}(2K(m) x \, | \, m)$, we have 
\begin{align*}
    \frac{(\sqrt{\rhot{m}})''}{\sqrt{\rhot{m}}}(x)
    &= K(m)^2 \left(2 \frac{\mathrm{dn}''}{\mathrm{dn}}-\left(\frac{\mathrm{dn}'}{\mathrm{dn}}\right)^2 \right)(2 K(m) x \, | \, m) \\
    &= 4 K(m)^2 \left(\frac{2-m}{4} - \frac{3}{4} \mathrm{dn}^2(2K(m) x\, | \, m) - \frac{m-1}{4 \mathrm{dn}^2(2K(m) x \, | \, m)}\right).
\end{align*}
Moreover,
\[
    \mathcal S(\rhot{m})
    =\frac{4 K(m)^2 \sqrt{1-m}}{\pi^2} \quad \text{and therefore} \quad \pi^2 \frac{\cS(\rhot{m})^2}{(\rhot{m})^2} = 4 K(m)^2 \frac{1-m}{\mathrm{dn}^2(2K(m) x  \, | \ m)}.
\]
Therefore, up to constants, substitution into~\eqref{eq:second_order_potential_main} yields
\begin{align*}
    v_{\rm TFW}(N\rhot{m})
    =&- 4 K(m)^2\biggr(  \left(N^2 - \frac14\right) \mathrm{dn}^2(2K(m) \, x \, |\, m) + \left(\frac{1}{12} + \theta_N\right) \frac{1-m}{\mathrm{dn}^2(2K(m)\, x \,| \, m)}
    \biggr).
\end{align*}
Finally, we can use the identities
\[
    \mathrm{dn}^2(u|m)=1-m\,\mathrm{sn}^2(u\, |\, m),
    \qquad
    \frac{1-m}{\mathrm{dn}^2(u|m)}
    =
    1-m\frac{\mathrm{cn}^2(u | m)}{\mathrm{dn}^2(u\, |\, m)}
\]
to obtain~\eqref{eq:Lame_inverse_potential}. The fact that this is exactly the Kohn-Sham potential follows from Theorem~\ref{thm:KS_high_density} and Proposition~\ref{prop:zero deficit}.

For the line case, using the definition
\[
    \rhor{\mu}(x)=\frac{1}{\pi} \frac{\mu}{\cosh(\mu  x)},
\]
a direct computation gives
\[
    \frac{(\sqrt{\rhor{\mu}})''}{\sqrt{\rhor{\mu}}}(x)
    = -\frac{\mu^2 }{2} + \frac{3\mu^2}{4} \frac{\sinh(\mu x)^2}{\cosh(\mu x)^2} = 
    \frac{\mu^2}{4}
    -\frac{3\mu^2}{4}\frac{1}{\cosh^2(\mu x)}.
\]
Since $3c_{\rm TF}=\pi^2$, it follows that
\[
    v_{\rm TFW}(N\rhor{\mu}) = -3c_{\rm TF}N^2(\rhor{\mu})^2
    -\frac13\frac{(\sqrt{\rhor{\mu}})''}{\sqrt{\rhor{\mu}}}
    =
    -\frac{\mu^2}{12}
    -
    \mu^2\left(N^2-\frac14\right)\frac{1}{\cosh^2(\mu x)}.
\]
In this case, we cannot show that $v_{\rm KS}(N\rhor{\mu}) = v_{\rm TFW}(N\rhor{m})$ via Lemma~\ref{lem:force_balance} as we do not have a priori existence of a potential. However, one can explicitly compute the first $N$ eigenfunctions of $-\Delta_\R + v_{\rm TFW}(N\rhor{\mu})$ and verify that their density is exactly $N\rhor{\mu}$. This proves~\eqref{eq:PT_inverse_potential}.
\end{proof}

\addtocontents{toc}{\protect\setcounter{tocdepth}{-1}}
\section*{Acknowledgements}

T.~C.~Corso acknowledges funding by the \emph{Deutsche Forschungsgemeinschaft} (DFG, German Research Foundation) - Project number 572811220. 

\section*{Declaration of generative AI use}
The authors acknowledge the use of generative artificial intelligence (AI) tools during the preparation of this manuscript, particularly for language editing, exploring and organizing mathematical arguments, and preliminary verification of bibliographic information. All mathematical proofs were carefully written and proofread by the authors, who take full responsibility for the content of the paper.

\addtocontents{toc}{\protect\setcounter{tocdepth}{2}}
\appendix

\section{Proof of Remark~\ref{rem:finite-rank}}
We now prove the finite rank optimality in Remark~\ref{rem:finite-rank}.

\begin{proof}[Proof of Remark~\ref{rem:finite-rank}]
Let $0\neq \gamma \in \mathfrak S(\R)$ with $\rank \gamma \leq N$, and write 
\[ \gamma = \sum_{j=1}^N \lambda_j \ketbra{\psi_j}{\psi_j}, \qquad 0\leq \lambda_j \leq 1 ,\qquad \inner{\psi_j,\psi_k} = \delta_{jk}.  \]
Let us first assume that each $\psi_j$ has compact support. Then, let $u\in C^\infty_c(\R;\R)$ with $\norm{u} = 1$ and, for each $\varepsilon>0$, choose $x_{j,\varepsilon}\in \R$ so large that the functions
\[u_j(x) \coloneqq \varepsilon^{\frac12} u(\varepsilon(x-x_{j,\varepsilon})) \]
have mutually disjoint support and have support disjoint from each $\psi_j$. Then, the operator
\[ \gamma_\varepsilon \coloneqq \sum_{j=1}^N \left|\sqrt{\lambda_j} \psi_j + \sqrt{1-\lambda_j} u_j \right\rangle \left\langle \sqrt{\lambda_j} \psi_j + \sqrt{1-\lambda_j} u_j \right| \]
belongs to $\mathfrak S_N(\R)$ and has rank $N$. Moreover, the disjoint support condition implies that
\[ \rho_{\gamma_\varepsilon}\geq \rho_\gamma, \quad T[\gamma_\varepsilon] =  T[\gamma] + \varepsilon^2 (N-\Tr \gamma) \norm{u'}_{\mL^2}^2, \quad \text{and}\quad G(\rho_{\gamma_\varepsilon}) = G(\rho_\gamma) + \varepsilon^2 (N-\Tr \gamma) \norm{u'}_{\mL^2}^2,\]
and therefore
\[ \frac{T[\gamma_\varepsilon] + 3 G(\rho_{\gamma_\varepsilon})}{\int \rho_{\gamma_\varepsilon}^3} \leq \frac{T[\gamma] + 3G(\rho_\gamma)}{\int \rho_\gamma^3}  +4\varepsilon^2 (N-\Tr \gamma) \frac{\norm{u'}_{\mL^2}^2}{\int \rho^3}.\]
Thus, since $\Tr \gamma_\varepsilon = N$, we can use Theorem~\ref{thm:LT_line} and take the limit $\varepsilon\downarrow 0$ to conclude that
\begin{align} c_N \leq \liminf_{\varepsilon \downarrow 0} \frac{T[\gamma_\varepsilon] + 3G(\rho_{\gamma_\varepsilon})}{\int \rho_{\gamma_\varepsilon}^3} \leq \frac{T[\gamma] + 3G(\rho_\gamma)}{\int \rho_\gamma^3}.\label{eq:compact support ineq} \end{align}

For general $\gamma$, we pick an approximating sequence $\psi_{j,n} \in C^\infty_c(\R)$ with $\psi_{j,n} \rightarrow \psi_j$ in $\mH^1(\R)$ and define
\[ \gamma_n \coloneqq \frac{\sum_{j} \lambda_j \ketbra{\psi_{j,n}}{\psi_{j,n}}}{\norm{\sum_{j} \lambda_j \ketbra{\psi_{j,n}}{\psi_{j,n}}}}.\]
The normalization guarantees that $0\leq \gamma_n \leq 1$ so that~\eqref{eq:compact support ineq} holds for $\gamma_n$. On the other hand, the orbital convergence implies that
\[ \lim_{n\rightarrow \infty} \frac{T[\gamma_n] + 3G(\rho_{\gamma_n})}{\int \rho_{\gamma_n}^3} =\norm{\gamma}^2 \frac{T[\gamma] + 3G(\rho_\gamma)}{\int \rho^3} \leq \frac{T[\gamma] + 3G(\rho_\gamma)}{\int \rho^3}. \] 
The result then follows from~\eqref{eq:compact support ineq} for each $\gamma_n$.
\end{proof}

\section*{Data availability}
No datasets were generated or analyzed during the current study.

\section*{Competing interests}
The authors have no competing interests to declare that are relevant to the content of this article.


\printbibliography
\end{document}